%% file: _main_journal_version2.tex
\documentclass[journal]{IEEEtran}
\usepackage{url}
\usepackage[utf8]{inputenc} 
\usepackage[nolist,nohyperlinks]{acronym}

\input{my_macros}

\usepackage{tabularx}
\usepackage{makecell}
\usepackage{graphicx}
\usepackage{booktabs}

\makeatletter
\AtBeginDocument
{
	\def\ltx@label#1{\cref@label{#1}}%add braces
	\def\label@in@display@noarg#1{\cref@old@label@in@display{#1}}%remove braces
	\def\label@in@mmeasure@noarg#1{%
		\begingroup%
		\measuring@false%
		\cref@old@label@in@display{#1}%remove braces for multline, see https://tex.stackexchange.com/q/737204/2388
		\endgroup}%  
} %
\makeatother

\usepackage{enumitem}
\usepackage{cases}

\newtheorem{lemma}{Lemma}
\newtheorem{definition}{Definition}
\newtheorem{proposition}{Proposition}
\newtheorem{assumption}{Assumption}
\newtheorem{theorem}{Theorem}

\newtheorem{remark}{Remark}

\newcommand{\Ccolor}{\color{black}}
\newcommand{\Cblack}{\color{black}}

\begin{document}
	\bstctlcite{MyBSTcontrol}
	\title{
		Joint Communication Scheduling and Resource Allocation for  Distributed Edge Learning:
		Seamless Integration in Next-Generation Wireless Networks
	}
	
	\author{
		Paul Zheng, Navid Keshtiarast, Pradyumna Kumar Bishoyi, Yao Zhu,\\ Yulin Hu, Marina~Petrova, and Anke~Schmeink
		\thanks{The work of Y. Hu was supported in part by the National Key R\&D Program of China under Grant 2023YFE0206600, National Natural Science Foundation of China under Grant 62471341, Hubei Provincial Science and Technology Cooperation Project under Grant 2025EHA040, and by the Fundamental Research Funds for the Central Universities under Grant 2042025KF0039.
        The work of P. Zheng, N. Keshtiarast, P. K. Bishoyi, Y. Zhu, M. Petrova, and A.
Schmeink is by the BMFTR Germany in the project ``Open6GHub"
under grant 16KISK012.
Y. Zhu and Y. Hu are the corresponding authors. Part of this work has been presented in IEEE ICC 2025~\cite{Zheng_conf_ICC_2025}.}
        
		\thanks{
			P. Zheng and A. Schmeink are
			  with Chair of Information Theory and Data Analytics, 
			  RWTH Aachen University, D-52074 Aachen, Germany,   
			(email: $zheng|schmeink$@inda.rwth-aachen.de). 
			}
		\thanks{N. Keshtiarast and M. Petrova are with Mobile Communications and Computing Group, RWTH Aachen University,
			D-52062 Aachen, Germany,   
			(email: $navid.keshtiarast|petrova$@mcc.rwth-aachen.de).}
            \thanks{P. K. Bishoyi is with Department of Electrical Engineering, Indian Institute of Technology Jodhpur, Rajasthan, India, 342030, (email: $pradyumna$@iitj.ac.in).}
		\thanks{
			Y. Zhu and Y. Hu are with School of Electronic Information, Wuhan
			University, 430072, China (email: $yao.zhu|yulin.hu$@whu.edu.cn).
			}
            \vspace{-1cm}
		}

	\IEEEoverridecommandlockouts
	\IEEEpubid{\begin{minipage}{\textwidth}\ \\[50pt]
			{\copyright 2025 IEEE. Personal use of this material is permitted.  Permission from IEEE must be obtained for all other uses, in any current or future media, including reprinting/republishing this material for advertising or promotional purposes, creating new collective works, for resale or redistribution to servers or lists, or reuse of any copyrighted component of this work in other works. Citation information: DOI 10.1109/TWC.2025.3635227}
	\end{minipage}} 
	
	\maketitle

\acrodef{DL}{distributed edge learning}
\acrodef{HB}{high bandwidth}
\acrodef{FL}{federated learning}
\acrodef{IoT}{Internet-of-Things}
\acrodef{CR}{communication round}
\acrodef{UE}{user equipment}
\acrodef{WFL}{wireless FL}
\acrodef{CS}{communication scheduling}
\acrodef{JCSRA}{joint communication scheduling and resource allocation}
\acrodef{RB}{resource block}
\acrodef{QoS}{quality of service}
\acrodef{BS}{base station}
\acrodef{MM}{Majorization-Minimization}
\acrodef{KKT}{Karush–Kuhn–Tucker}
\acrodef{LHS}{left hand side}
\acrodef{LP}{linear programming}
\acrodef{DVFS}{dynamic voltage and frequency scaling}

	\begin{abstract}
		Distributed edge learning (DL) is considered a cornerstone of intelligence enablers, since it allows for collaborative training without the necessity for local clients to share raw data with other parties, thereby preserving privacy and security. 
		Integrating DL into the 6G networks requires a coexistence design with existing services such as high-bandwidth (HB) traffic like eMBB. Current designs in the literature mainly focus on communication round-wise designs that assume a rigid resource allocation throughout each communication round (CR). 
        However, rigid resource allocation within a CR is a highly inefficient and inaccurate representation of the system's realistic behavior, especially when CR duration far exceeds the channel coherence time due to large model size or limited resources. 
        This is due to the heterogeneous nature of the system, as clients inherently may need to access the network at different time instants.  
        This work zooms into one arbitrary CR, and demonstrates the importance of considering a time-dependent design for sharing the resource pool with HB traffic. 
        We first formulate a time-slot-wise optimization problem to minimize the consumed time by DL within the CR while constrained by a DL energy budget. 
        Due to its intractability, a session-based optimization problem is formulated assuming a CR lasts less than a large-scale coherence time. 
        Some scheduling properties of such multi-server joint communication scheduling and resource allocation framework have been established.
		An iterative algorithm has been designed to solve such non-convex and non-block-separable-constrained problems. 
        Simulation results confirm the importance of the efficient and accurate integration design proposed in this work.
	\end{abstract}

	\IEEEpeerreviewmaketitle

	\vspace{-.25cm}
	\section{Introduction}
	\vspace{-.1cm}
	Data-driven machine learning techniques have shown substantial potential in tackling complex problems that are challenging for traditional optimizations, largely due to the abundance of data and increased computational capabilities.
	However, the vast quantity of data 
    collected or generated by the ever-increasing and diverse range of \ac{IoT} devices may contain sensitive and private information about individuals or industries. Sharing this raw information with other parties poses inevitable concerns.
	
	To address privacy and security concerns, a line of \ac{DL} frameworks, such as \ac{FL}~\cite{mcmahan_FL_2017} and split learning~\cite{Gupta_2018_split_learning}, has been proposed. In such frameworks, clients are not required to transmit their data to any third party; instead, only model parameter or intermediate activation output has to be communicated.
	Such \ac{DL} has drawn immense attention in the past few years and is being considered for integration (or already is integrated) into real-world mobile/edge applications~\cite{Li2020_reviewFL_applications}, e.g., Google Gboard~\cite{hard_2019_FLmobilekeyboard}. 
	Compared to within-cluster \ac{DL}, several challenges arise in this \ac{DL} scenario~\cite{LiTian_FL_hetero_survey_2020}: i)~the communication latency is far from negligible~\cite{Zhu_delayedGradientAveraging_2021}; ii)~heterogeneity in terms of data distribution and quantity and system itself (communication latency, computational capacity, etc.). 
	One of the key goals of the research community has been to \emph{enable faster DL training while maintaining good accuracy}.
	Significant research has been devoted to gradient or model update compression techniques to reduce the communication/computation load~\cite{Banerjee_FLCompression_2021, Sattler_RobustCommEfficient_FL_Compression_2020, Chiappa_FedPAQ_Quantization_2020, Shlezinger_VectorQuantizationFL_2021}. 
	Even with the significantly reduced communication payload, the training time remains highly impacted by the communication delays. A substantial portion of the training latency persists due to the synchronous property of \ac{DL} and the inherent system heterogeneity. For instance, in synchronous FL, the server needs to wait until it receives all selected clients' updates before aggregating them to a new global model for the next \ac{CR}. As a result, the \ac{DL} training latency is determined by the slowest \acp{UE}, referred to as stragglers~\cite{LiTian_FL_hetero_survey_2020}. 
	A line of research has focused on asynchronous DL, which relaxes the need to wait for all client updates~\cite{Xu_AsyncFL_Survey_2023} by proceeding the model aggregation after a certain time regardless of missing some of clients updates. 
    However, even in asynchronous DL, the fastest users still experience idle times while communicating and receiving model updates. An even more efficient asynchronous framework called overlapped communication and computing FL has been proposed~\cite{Zhu_delayedGradientAveraging_2021, Chen_FLoverlapComputationCommunication_2023}. This framework allows \acp{UE} to continue their local updates while communicating their model updates. Correction terms are added to adjust the training performance.
However, asynchronous distributed edge learning frameworks often overlook the practical constraints of scarce wireless bandwidth and limited energy budget of wireless edge devices. In wireless settings, specific design criteria must therefore be considered~\cite{Lim_FL_MEC_Survey_2020}. 

	In \ac{WFL}, numerous works have concentrated on reducing the overall training latency while allocating both training and wireless resources~\cite{ChenMingzhe_ConvTimeOptiFL_2021,  Zhu_LatencyMinWFL_2024, Shi_DeviceSchedulingResourceAllocationWFL_2021, Liu_TrainingTimeMinFEELBandwidth_2022,Luo_ClientSamplingFL_2022, Xu_MultiFLServiceBandwidthAllocation_2022, Vu_MinExcecutionTimeFLCellFreeMMIMO_2022, Chen_RelayWFLOptimization}. Probabilistic user selection and~\ac{RB} allocation design has been proposed in~\cite{ChenMingzhe_ConvTimeOptiFL_2021, Zhu_LatencyMinWFL_2024, Shi_DeviceSchedulingResourceAllocationWFL_2021} for accelerating the FL training latency. Liu et al.~\cite{Liu_TrainingTimeMinFEELBandwidth_2022} jointly design the quantization level and bandwidth allocation.  Similar designs have been extended to more complex scenarios, e.g., in multiple coexisting FL services~\cite{Xu_MultiFLServiceBandwidthAllocation_2022}, cell-free massive MIMO~\cite{Vu_MinExcecutionTimeFLCellFreeMMIMO_2022} and relay-assisted networks~\cite{Chen_RelayWFLOptimization}.

The above works however do not take into consideration the practical energy limitations of local devices. To address this issue,  various resource allocation designs for \ac{WFL} schemes have been proposed. 
	Tran et al.~\cite{Tran_WFL_Optimization_2019} established the foundation for an optimization model that considers both time and energy constraints in \ac{WFL} systems. Their research presents a joint computational capacity and power allocation for optimizing the energy and time spent. Other studies have extended the design to a variety of scenarios, including hierarchical WFL~\cite{Feng_MinMaxOptiHWFL_2022}, multi-cell~\cite{Nguyen_2022_LatencyOpti_blockchainFL, Wu_2024_CostEfficientFLMultiCell}, wireless powered networks-enabled FL~\cite{Poposka_2024_DelayMinFLWPT},  neural network partition split at server and client sides~\cite{Deng_LowLatencyWFL_DNNPartition_2023}, and hybrid local and centralized learning~\cite{Huang_WFLHybridLocalCentral_2023}. 
	Yang et al.~\cite{yang_FLenergy_2021}, building upon~\cite{Tran_WFL_Optimization_2019}, incorporated local training accuracy as an additional design variable.  
In~\cite{Xu_2021_LongTermFL}, a joint user scheduling and bandwidth allocation strategy for FL is proposed, taking into account the increased significance of updates throughout the training. 
	The aforementioned studies demonstrate the trade-off between learning accuracy, latency and energy, emphasizing the necessity for an accurate characterization of their interdependencies.

	However, all the above works assume rigid resource allocation that remains constant throughout each \ac{CR} of DL. If we zoom into an arbitrary \ac{CR}, it is evident
	that there may exist communication periods with lower communication requirements (downlink communication or when only a few clients are ready to perform uplink transmission), as well as periods that demand significantly higher bandwidth, e.g., when multiple UEs simultaneously need to update their local model updates.
	Depending on the heterogeneity of local computational task load and speeds, clients transmit their models back to the server at different times. It is illogical and inefficient to pre-allocate resources to a UE for the whole duration of \ac{CR} if such resources are not needed in a certain time slot{, where each time slot represents a schedulable time interval. This statement holds unless the neural network model size is so small that the duration of a CR including local computation and communication fits within one schedulable interval}. 
	This work will focus therefore not only on the resource allocation, but also on the \emph{\ac{CS}}  problem~\cite{Li_FLCommScheduling_2021} for FL tasks when a CR spans multiple schedulable intervals.  It is important to clarify that the scheduling here does not mean selecting the client for participating in a CR, but rather the time when the \ac{UE} can be ready to do its uplink transmission updates. 
    This time can be controlled by adjusting the computational capacity.
While aforementioned works also tune the computational capacity, they overlook the impact of \acp{CS} due to rigid wireless resource allocation schemes.

    There are several works that have studied the importance of such CS problems. For instance, \cite{OzfaturaGundunz_OverlapFEEL_2021} evaluated the scheduling strategy based on 
	the minimum remaining time of updates; Luo et al.~\cite{Luo_CostEffectiveWFL_2021} derived an optimal CS ordering based on local computation time, only applicable when the downlink communication is ignored. Furthermore, authors in~\cite{Li_FLCommScheduling_2021} establish an optimal CS structure and design jointly optimal batch and user selection. 
    For enhancing performance, the resource allocation can be jointly designed with CS. 
    Xu~\cite{Xu_2024_TDMA_WFL} considers \ac{JCSRA}
    in the TDMA scheme given a session ordering, formulated with an optimization problem deciding computational capacity to minimize overall latency.   
    However, all the existing CS literature exclusively addresses \emph{single-server systems}, which we define as systems where only one client can perform uplink transmission at a time, borrowing the terminology from queuing theory where each client acts as a ``server'' processing its transmission task. This is shown to be only optimal when the rate is assumed to be linear with the number of allocated resources, which does not hold under device transmit power constraint.  
    We explore the generalization of such CS to OFDMA, when multiple \acp{RB} are available for multiple clients to perform the uplink transmission simultaneously. 
    We classify our work as a problem of \emph{multi-server \ac{JCSRA}}, which we define as a system where multiple \acp{UE} can transmit simultaneously, with the same borrowed terminology.
    To the best of our knowledge, it has only been explored in the case of massive MIMO by~\cite{Vu_2022_sessionBasedMIMOFL_designs}, not in other multi-access schemes. 
    In contrast to~\cite{Vu_2022_sessionBasedMIMOFL_designs}, our work provides a time-slot-wise formulation and coexistence design with existing network services that directly motivate \ac{JCSRA}. While~\cite{Vu_2022_sessionBasedMIMOFL_designs} focuses on MIMO, we target JCSRA under OFDMA, leading to fundamentally different design considerations. Furthermore, we establish deeper theoretical results on the problem's scheduling properties and feasibility, and we present more comprehensive simulation results that validate JCSRA's performance and offer richer interpretations of JCSRA.

Given the more efficient and also more accurate representation of DL provided by~\ac{JCSRA}, it is natural to extend the discussion to its coexistence with other wireless services such as \ac{HB} traffic (e.g., eMBB) and URLLC.
	Beyond theoretical considerations, it is also a practical concern since DL is expected to be integrated into the future-generation networks as already considered by 3GPP~\cite{3GPP_AIML_service}.
It is essential to emphasize that, in contrast to other services where the \ac{QoS} requirements for traffic are generally characterized by real-time demands and are therefore not ``controllable," DL service traffic can be, as mentioned in previous paragraph, tunable with CS.
	There exists little literature that studies DL coexistence with other services~\cite{Ganjalizadeh_DeviceSelectionURLLCDistributedLearning_2022,FL-nonFL_Farooq_2024,Li_FL_BandwidthSlicing_opticalNetworks_2020, Lin_CFLITCoexistingFL_2023}. In~\cite{Ganjalizadeh_DeviceSelectionURLLCDistributedLearning_2022}, the authors investigate the integration of DL and URLLC services in industrial networks, proposing a risk-sensitive device selection, aimed at minimizing DL training delay while ensuring URLLC QoS.
    Further, in~\cite{FL-nonFL_Farooq_2024}, the coexistence of FL and HB traffic is examined under half- and full-duplex massive MIMO schemes.
	The work~\cite{Li_FL_BandwidthSlicing_opticalNetworks_2020} considers bandwidth slicing in optical networks for FL and non-FL users. 
	Lin et al.~\cite{Lin_CFLITCoexistingFL_2023} investigate co-existing over-the-air FL with other information transfers. 
	However, time-dependent resource allocation has also been largely overlooked in the coexistence literature, leading to a mismatch with realistic system behavior and an inefficient use of resources.
	
	In this work, we investigate \ac{JCSRA} for the coexistence of HB traffic services and DL within {an arbitrary} CR for any given client selection. The proposed JCSRA framework operates within a CR and does not alter the FL update rule or convergence behavior.
	For the sake of modeling clarity, this work will use vanilla FL as a representative example of DL.
	Any synchronous DL framework can be adapted using the same design principles. {In contrast, 
	asynchronous or overlapping DL frameworks could also benefit conceptually, but incorporating them would require additional scheduling granularity and unmanageable optimization complexity.
	} The aim of the work is to demonstrate the necessity and the hidden complexity of the \emph{multiple-server JCSRA} problem,
	along with the potentially large gap and completion time estimation error that can occur compared to the rigid allocation and single-server JCSRA. 
    The key contributions of this work are as follows:
	\begin{itemize}[leftmargin=11pt]
		\item To the best of our knowledge, this work is the first formulation of a time-slot-wise resource allocation and computational speed optimization problem aimed at minimizing the end latency of a CR of vanilla FL under an energy budget constraint when coexisting with HB traffic. The formulation captures the full complexity of the system and motivates the necessity of JCSRA schemes.
		\item To address the intractability of the {time-slot-wise problem,} we assume a large-scale coherence time, which enables us to optimize w.r.t. the average channel information due to the potentially large model size to be communicated. From the assumption, we propose an equivalent tractable session-based optimization framework
		that jointly controls the downlink and uplink duration of each session, allocation of RBs of the coexisting DL and HB traffic within each session, and the computational capacity under energy budget constraints.
		Necessary and sufficient feasibility conditions have been established.
		\item 
			We identify and theoretically prove that the optimal JCSRA system is in general not in a single-server form (sequential transmission). Furthermore, we establish that the non-preemptive and non-idle properties remain in the multi-server case (where multiple clients can transmit simultaneously).
		
		\item To tackle the non-convex and non-block separable constrained problem, we propose an iterative algorithm that solves a convex sub-problem in each step, thereby ensuring convergence to a stationary point. Additionally, we introduce a reasonable ranking heuristic, which is validated by simulations as being effective.
		\item The simulation results first confirm the convergence and good performance level of the heuristic ordering. They also demonstrate the efficiency of the JCSRA design. In resource-constrained systems, JCSRA methods provide significant latency improvement compared to rigid allocation. Additionally, factors that increase the performance gap between single-server and multi-server JCSRA are identified. 
		{An example integration of the JCSRA results into a real-time algorithm shows that the predicted performance remains achievable under fast fading with integer RB allocation.
		}
	\end{itemize}
\input{tables/notation_table}

\vspace{-.32cm}
\section{System Model and Initial Problem Formulation}
\vspace{-.08cm}

Consider a 5G NR system with TTI slots of lengths~$\Delta$\,(s){, where each time slot represents a schedulable time interval for resource allocation decisions}.
It consists of a single cell with HB traffic UEs~$\setE$ using OFDMA. The service of \ac{DL} is expected to be integrated, we consider~$\setF$ the \acp{UE} that participate in FL. In total~$K$ \acp{RB}  are available for both services. 
The bandwidth of one RB is denoted as~$B$\,(Hz).

 \vspace{-.35cm}
\subsection{HB Traffic UEs (e.g., eMBB)}
 \vspace{-.08cm}
Each HB traffic UE $e\in\setE$ performs downlink transmission with~$K_e^{(t)}$ RBs at time slot~$t$. For the ease of notation regarding the actual contribution of this work, we assume from here the channels are frequency-flat, which will be justified later in~\Cref{remark: frequency-flat}. The rate expression can be written as:
\begin{equation}
	r^{(t)}_{e} = K_e^{(t)}B\log_2\Big(1+\gamma_{e}^{(t)} \Big) ,
\end{equation}
where~$\gamma_{e}^{(t)}= \frac{P^{(\mathrm{dl})}h_{e}^{(t)2}}{BN_0}$ is the SNR of HB traffic UE~$e$; $h_e^{(t)}$ the channel coefficient of UE~$e$ and $N_0$ the AWGN noise spectral density. The \ac{BS} is assumed to have a constant downlink power~$P^{(\mathrm{dl})}$ at each RB.

A fair rate allocation for HB traffic UEs needs to be ensured. The requirement is for all HB traffic UEs to have the time average rate above a threshold~$\theta$\,(bit/s):
\begin{equation}
	\min_{e\in\setE}\frac{1}{T}\sum_{t=1}^{T} r_e^{(t)} \geq \theta,
	\label{eq: time_raw_eMBB}
\end{equation}
where~$T>0$ is the ending time of the considered CR. 

\vspace{-.15cm}
\begin{remark}
\label{remark: frequency-flat}
The long-term average channel states are typically assumed to be frequency-flat. Since from section~\ref{sec: rigid} throughout this work, only long-term channel states are relevant, the {time-slot-wise} channel was defined likewise to avoid unnecessary redefinitions.
\vspace{-.15cm}
\end{remark}

\vspace{-.3cm}
\subsection{Vanilla Federated Learning}
\vspace{-.1cm}
To highlight the benefit of time-dependent resource allocation, without loss of generality, this work focuses on vanilla FL, while the proposed approach can be applied to any other advanced FL framework.
The following stages that consist of a CR are iteratively performed in vanilla FL, i.e., FedAvg~\cite{mcmahan_FL_2017}: 
\begin{enumerate}[leftmargin=14pt]
	\item The BS randomly (or with any client selection) selects~$S$ UEs~$\setS\subset\setF$ to participate in the CR of FL training, and broadcasts the current global model to selected UEs. 
	\item After receiving the global model, each UE~$s\in\setS$ trains for $I_s$ epochs with its own local dataset.
	\item Each UE sends back the locally trained model once the local training is done.
	\item After receiving all model updates from UEs, the average of the model updates is computed at the BS and is considered as the global model for the next CR.
\end{enumerate}
In this work, we focus on each CR. We assume therefore an arbitrary client selection~$\setS\subset\setF$ and the step~$4)$ is ignored since it is not impacted by wireless resource allocations. We denote $\setS=\{1,\ldots,S\}$.

\noindent\textbf{Downlink Phase:} FL downlink communication uses fountain-coded multicasting~\cite{Castura_RatelessCode_2006} as assumed in~\cite{OzfaturaGundunz_OverlapFEEL_2021}. 
The downlink rate of FL UE~$s\in\setS$ at {time slot}~$t$ is:
\begin{equation}
	r^{(t)}_{s,\mathrm{dl}} = K_{\mathrm{dl}}^{(t)}B\log_2\Big(1+P^{(\mathrm{dl})}\gamma_{s}^{(t)} \Big),
\end{equation}
where~$K_{\mathrm{dl}}^{(t)}$ the total number of RBs given to the downlink broadcasting at time~$t$; 
$\gamma_{s}^{(t)} = \frac{h_s^{(t)}}{BN_0}$ the downlink power-normalized SNR of UE~$s$, i.e., SNR per transmit power,
 with $h_s^{(t)}$ being the channel gain of UE~$s$ at {time slot}~$t$.

For a model of size~$D$ bits, the downlink communication for UE~$s$ will last until the complete model has been uploaded, defining the downlink duration $\tau_{s,\mathrm{dl}}$:
\begin{equation}
(\forall s\in\setS)\quad	\Delta\sum_{t=1}^{\tau_{s,\mathrm{dl}}} r_{s,\mathrm{dl}}^{(t)} \geq D.
	\label{eq: time_raw_dlcompletion}
\end{equation}

\noindent\textbf{Local Training Update:} Each UE asynchronously starts the local training independently after correctly decoding the whole model. The training duration~$\tau_{s}^{(\mathrm{cp})}${ is determined and assigned to each UE. UEs attempt to meet as close as possible this delay requirement by tuning} the computational capacity~$f_{s}\in(0, f_{s,\max}]${ by adjusting chip voltage with the technique of \ac{DVFS}~\cite{Chandrakasan_EnergyCMOS_1992, Zhai_DVFS_2004, Tang_DVFS_GPU_2019}. The latency and energy can be expressed as follows~\cite{yang_FLenergy_2021, Mao_MECEnergyModel_2016, PK_TGCN, Zeng_energyCPUGPUFL_2021}}:
\begin{equation}
(\forall s\in\setS) \quad	\tau^{(\mathrm{cp})}_{s} = \frac{I_sC_s\Theta_s}{f_{s}} = \frac{\zeta \Theta_s}{f_s},
	\label{eq: tau_cp}
\end{equation}
where~$C_s$ (cycles/sample) the number of CPU cycles required for training one sample data at UE~$s$; $I_s$ the number of local epochs and $\Theta_s$ is the local data sample size. {We define~$\zeta\defeq I_sC_s$ assumed identical across $s\in\setS$.}
The energy consumed on the local computation by UE~$s$ writes as:
\begin{equation}
(\forall s\in\setS) \quad	E^{(\mathrm{cp})}_{s} = \kappa I_s C_s \Theta_sf_{s}^2,
	\label{eq: E_cp}
\end{equation}
with~$\kappa>0$ is the effective switched capacitance~\cite{Chandrakasan_EnergyCMOS_1992}. 
{Since the computation time scales inversely with the computational capacity~$f_s$, while the computation energy scales with its square, computing more slowly  at a lower capacity can substantially reduce the total energy consumption for a given computational load. This reflects the so-called \emph{energy-delay tradeoff}.}

\begin{remark}
{The relationship between computational frequency, latency, and energy consumption has been shown to be more complex in realistic environment, due to concurrent device services~\cite{Li_SmartPC_2019} and GPU-CPU-memory coordination in deep learning tasks~\cite{Guo_BoFL_2022}. Nevertheless, the simplified model provides valuable insights as long as such an energy-delay tradeoff exists. Future work will explore black-box optimization techniques to adapt to the unknown frequency-performance relationship in realistic settings.}
\end{remark}

\noindent\textbf{Model Uplink Update Phase:} After the local training, each UE, on its own, independently requests to transmit the updated model to BS for averaging. 
Unlike in downlink transmission where the BS has the power to serve each RB with sufficient transmit power capacity,
UEs have limited transmit power. As more RBs are allocated to a UE, less power can be allocated per RB.
Assuming that each RB can be shared and split by UEs over a long-time scale, the rate can be derived from~\cite{Shen_OFDM_allocation_2005}:
\begin{equation}
(\forall s\in\setS) \quad	r_{s,\mathrm{ul}}^{(t)} = K_{s,\mathrm{ul}}^{(t)} B \log_2\Bigg(1+\frac{p_{s,\mathrm{ul}}^{(t)}\gamma_{s}^{(t)}}{K_{s,\mathrm{ul}}^{(t)}}\Bigg),
	\label{eq: uplink rate (time raw)}
\end{equation}
where~$K_{s,\mathrm{ul}}^{(t)}$ the number of RBs used by UE~$s$ for the uplink transmission at TTI~$t$;  $p_{s,\mathrm{ul}}^{(t)}$ the transmission power used for uplink transmission; $\gamma_{s}^{(t)}$ as defined previously. Each UE~$s$ is subject to a maximum transmit power $P_{\max}$, i.e., $p_{s,\mathrm{ul}}^{(t)}\in[0, P_{\max}]$.
Similar to downlink, the uplink communication duration~$\tau_{s,\mathrm{ul}}$ is characterized by UE~$s$ completing its transmission of~$D$ bits of model update:
\vspace{-.05cm}
\begin{equation}
(\forall s\in\setS) \quad	\Delta\sum_{t=\tau_{s,\mathrm{dl}}+\tau_s^{(\mathrm{cp})}+1}^{\tau_{s,\mathrm{dl}}+\tau_s^{(\mathrm{cp})}+\tau_{s,\mathrm{ul}}} r_{s,\mathrm{ul}}^{(t)} \geq D,
	\label{eq: time_raw_uplink_completion}
\end{equation}
where the uplink transmission starts after the completion of both downlink communication and local computation.
The resulting consumed energy in the uplink communication is:
\vspace{-.35cm}
\begin{equation}
(\forall s\in\setS) \quad	E_s^{(\mathrm{cm})} =\Delta \sum_{t=\tau_{s,\mathrm{dl}}+\tau_s^{(\mathrm{cp})}+1}^{\tau_{s,\mathrm{dl}}+\tau_s^{(\mathrm{cp})}+\tau_{s,\mathrm{ul}}} p_{s,\mathrm{ul}}^{(t)}.
\end{equation}
The overall \ac{CR} latency  is characterized by the slowest UEs:
\begin{equation}
	T = \max_{s\in\setS}\{ \tau_{s,\mathrm{dl}} + \tau^{(\mathrm{cp})}_{s} + \tau_{s,\mathrm{ul}}\}.
\end{equation}

The overall consumed energy of the \ac{CR} for UE~$s$ is denoted as $
E_s^{(\mathrm{tot})} = E_s^{(\mathrm{cp})}+E_s^{(\mathrm{cm})}.$ Each UE~$s$ has an FL training energy budget of~$E_{s, \mathrm{budget}}$ or the whole FL system can be subject to a network-wide energy budget~$E_{\mathrm{budget}}$.

\subsection{General {Time-Slot-Wise} Problem Formulation}
\label{Section: time-step-wise problem formulation}
We aim at efficiently allocating limited RBs to both HB traffic and FL traffic 
 along with
the uplink transmission power and the device computational capacity, to minimize the total latency of one FL \ac{CR} while satisfying energy budget constraints and coexisting HB traffic requirements. 
The problem can be formulated as follows:
\begin{subequations}
	\label{pb: orig_pb}
	\begin{align}
		\!\!\!\!	\min_{\{K_{e}^{(t)}, K_{s,\mathrm{dl}}^{(t)}, K_{s,\mathrm{ul}}^{(t)}, p_{s,\mathrm{ul}}^{(t)},f_s\}_{\forall e,s,t}, } \!\!\!\!\!\!\!\!\!\!\! & T , \\
		\mathrm{s.t.} \quad\quad\quad& \eqref{eq: time_raw_eMBB}, \eqref{eq: time_raw_dlcompletion},\eqref{eq: time_raw_uplink_completion} \notag\\
		&\hspace{-2.5cm} (\forall t),\ \sum_{e\in\setE}K_e^{(t)} + K_{s,\mathrm{dl}}^{(t)} + \sum_{s\in\setS}K_{s,\mathrm{ul}}^{(t)} \leq K,
		\label{cons: Res_sharing (t step)}\\[-.22cm] 
		& \hspace{-2.5cm}(\forall s\in\setS),\  E_s^{(\mathrm{tot})} \leq E_{s, \mathrm{budget}}, \ p_{s,\mathrm{ul}}^{(t)}\in[0, P_{\max}]\\[-.15cm] 
		&\hspace{-3.4cm} (\forall s\in\setS,e\in\setE,t), \ f_s\in (0, f_{s,\max}], \ K_{e}^{(t)}, K_{s,\mathrm{dl}}^{(t)}, K_{s,\mathrm{ul}}^{(t)}\in\N.
	\end{align}
	\vspace{-.1cm}
\end{subequations}
Constraint~\eqref{cons: Res_sharing (t step)} denotes{ that a total of~$K$ RBs can be allocated.}

\vspace{-.15cm}
\begin{remark}
	The problem is intractable due to several factors: the time-dependency of the starting time of each \ac{UE}'s uplink and each \ac{UE}'s finishing time (either for uplink and downlink), which depends on the past and future {optimization decisions (i.e., resource allocation and computational capacity assignments)}; the non-convexity; completion time (as detailed in the next remark) and therefore high time-dependence of variable dimension (number of time steps needed).
	\vspace{-.15cm}
\end{remark}
\vspace{-.15cm}
\begin{remark}
In the given formulation, the constraint of \textbf{completion time} with time-varying rate (otherwise $T=D/r$) is known to be challenging to address and has been considered in either UAV wireless networks~\cite{Zeng_CompletionTimeUAV_2018, Zhan_CompletionTimeUAV_2019, Yuan_CompletionTimeUAV_2023}, or in scheduling literature as the flow time (or makespan)~\cite{Im_flowtimeLknorm_2015, Xu_JobSchedulingResourcePacking_2021}. There is no general optimization strategy. Typical techniques involve finding a certain structure of the solution space, but the approach becomes infeasible with a larger number of UEs (compared to 1 UAV in UAV literature).
In scheduling literature, various strategies are developed and proved to be constant-competitive for optimal online scheduling in a simplified situation, but they do not apply to more complex optimization problems with more practical constraints. 
\vspace{-.15cm}
\end{remark}
\begin{remark}
The complexity of the optimization problem is also linked to the specificity of DL traffic. 
Despite its burstiness~\cite{luangsomboon2023burstiness} and high communication demand, the starting and finishing times of each traffic are controllable and only the end latency of CR matters instead of individual packet, in contrast to other common service traffic such as eMBB or URLLC, where the starting time, or packet arrival time, is not controllable and each packet to be transmitted within a certain latency constraint.
This tunability allows for greater efficiency potentials in resource optimization; for instance, if future congestion is anticipated, it may be beneficial to extend the local training (therefore transmission), to save more energy{ due to the energy-delay tradeoff in~\eqref{eq: tau_cp} and \eqref{eq: E_cp}} or expedite current transmissions to prevent future congestion.
\end{remark}
The above remarks highlight the complexity of such problems. To alleviate the strong time dependency of the current {optimization decisions} on the future, we make the following assumption. 
\vspace{-.15cm}
\begin{assumption}
	The CR happens within a large-scale channel coherence time, i.e., the channel statistics stay stationary during the CR.
	\label{assumption: large-scale coherence time}
	\vspace{-.15cm}
\end{assumption}
Based on this assumption, several designs can be proposed to mitigate the highly time-dependent nature of the problem.

\vspace{-.35cm}
\subsection{Rigid Resource Allocation (as Baseline)}
\label{sec: rigid}
\vspace{-.15cm}

Given the stationary channel assumption, most current WFL designs focus on 
a rigid RB allocation within each CR. 
Throughout the rest of the work, $\gamma_s${, $\gamma_e$} denotes{ respectively} the constant statistical average of power-normalized SNR~$\gamma_s^{(t)}${,~$\gamma_e^{(t)}$}.
The rigid formulation~$(\setP_{\mathrm{rig}})$ for solving the problem~\eqref{pb: orig_pb} is as follows with the variables set $\setX_{\mathrm{rig}}=\{K_{\mathrm{dl}}\in[0,K'], (K_{s,\mathrm{ul}}\in\R_+, p_{s,\mathrm{ul}}\in [0, P_{\max}], \tau_{s,cp}\in[{\tau_{s,\min}^{(\mathrm{cp})}},+\infty])_{s\in\setS}\}$: 
\begin{subequations}
	\begin{align}
		\min_{\setX_{\mathrm{rig}}} \quad & \max_{s\in\setS} \Big\{\frac{D}{r_{s,\mathrm{dl}}}\!+\! \tau_{s,cp} \!+\! \frac{D}{r_{s,\mathrm{ul}}}\Big\} ,
		\\[-.1cm]
		\mathrm{s.t.}\quad&  \sum_{s\in\setS} K_{s,\mathrm{ul}}\leq K', \label{cons: uplink RB sharing (rigid)} \\[-.3cm]
		&\min_{s'\in\setS}\Big\{\frac{D}{r_{s',\mathrm{dl}}}+ \tau_{s',cp}\Big\} \geq \max_{s'\in\setS}\Big\{\frac{D}{r_{s',\mathrm{dl}}}\Big\},
		\label{cons : dl/ul separation (rigid)}
		\\[-.1cm]
		& (\forall s\in\setS) \ \Big[\kappa \frac{\zeta^3\Theta_{s}^3}{\tau_{s,cp}^2} + \frac{p_{s,\mathrm{ul}}D}{r_{s,\mathrm{ul}}} \Big] \leq E_{s,\mathrm{budget}},
		\label{cons: energy (rigid)}
	\end{align}
	\label{pb: rigid problem}
%	\vspace{-.5cm}
\end{subequations}
where~$K'>0$ is the remaining amount of RBs when a constant minimum amount of RBs are allocated to HB UE service to ensure the constraint~\eqref{eq: time_raw_eMBB}.
Note that the variables of computational capacity~$\{f_s\}_{s\in\setS}$ are replaced equivalently by~$\{\tau_{s,cp}\}_{s\in\setS}$ by the one-to-one relationship given in~\eqref{eq: tau_cp};  the upper bound~$f_{s,\max}$ is transformed to~$\tau_{s,\min}^{(\mathrm{cp})}$. 
The constraint~\eqref{cons : dl/ul separation (rigid)} represents the separation (in time) of downlink and uplink phases (that will be motivated in the next section). 
Given tight energy constraint, it is possible that there is not enough energy to complete the local training and the model update transmission. The following feasibility condition is established.
\vspace{-.1cm}
\begin{theorem}[Feasibility Condition]
	The rigid problem~$(\setP_{\mathrm{rig}})$ is feasible if and only if $D/(B\gamma_s)<E_{s,\mathrm{budget}}$ and $K>a\theta${, with $a\defeq\sum\limits_e\frac{1}{B\log(1+\gamma_e)}$}.
	\label{theorem: feasibility (rigid)}
	\vspace{-.1cm}
\end{theorem}
\begin{proof}
     In Appendix~\ref{appendix: proof feasibility rigid}.
    \vspace{-.2cm}
\end{proof}
As mentioned in the introduction, the rigid-based allocation is a highly inefficient and inaccurate representation of system behavior under reasonable allocation.
This allocation serves as a baseline to evaluate the proposed session-based approach.

\vspace{-.3cm}
\section{Session-based RB allocation}
\vspace{-.15cm}
\subsection{Motivation}
\vspace{-.15cm}

\begin{figure}[t]
	\centering
	\subfloat[Heterogeneous\label{fig: illu_hetero}]{
		\includegraphics[width=0.49\linewidth, trim=5 0 0 5]{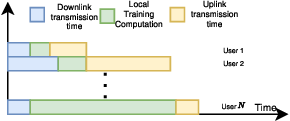}}
	%    \hfill 
	\subfloat[Homogeneous\label{fig: illu_homo}]{
		\includegraphics[width=0.49\linewidth, trim=5 0 0 5]{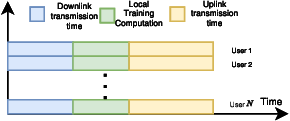}}
		\vspace{-.15cm}
	\caption{Example Illustration of the System Time-wise RB allocation for homogeneous and heterogeneous systems}
	\label{fig: illustration_heterogeneous}
	\vspace{-.65cm}
\end{figure}

The rigid resource allocation across the whole CR seems efficient only when considering solely DL services in a homogeneous network without energy constraint, where all UEs have identical channel strength, computational capacity, and computational load. In this case, a rigid allocation across the CR can be an accurate representation and efficient allocation strategy since the wireless communication resources are indeed shared simultaneously by all UEs, e.g., see an illustration in~\Cref{fig: illu_homo}.
When the system is heterogeneous, i.e., when local dataset sizes or the computational capacities vary significantly, the timing of uplink UEs ready to transmit their local updates can be staggered. 
As an example in~\Cref{fig: illu_hetero}, when UE~1 is ready to initiate uplink transmission, no other UEs are prepared to transmit; consequently, it can utilize significantly more RB than what is allocated rigidly. 
The design is more complex when the energy consumption is taken into account. If the network is not requested by other UEs immediately after UE~1 completes its transmission, it may be advantageous for UE~1 to proceed at a slower pace to conserve energy{ due to the energy-delay tradeoff in~\eqref{eq: tau_cp} and \eqref{eq: E_cp}}. If many UEs are expected to request network access shortly thereafter, UE~1 should aim to complete its tasks quickly before others begin their uplink transmissions to avoid further congestion.
To account for the time dependency of the{ shared RB pool} scheme, still under~\Cref{assumption: large-scale coherence time}, we propose to reformulate the problems with ``sessions".

\vspace{-.45cm}
\subsection{Session Definition and System Remodeling}
\vspace{-.1cm}
In general, resource allocation in wireless systems must be dynamically adjusted in response to the arrival or end of traffic demands. In the context of WFL, downlink broadcasting for all UEs is initiated simultaneously. Different downlink sessions are therefore characterized by each UE successfully fully decoded the broadcasted model. The uplink session, on the other hand, may involve a UE ready for uplink transmission, i.e. finalizing its local training,  or completing its uplink transmission.
Each session configuration needs to be handled separately since which UEs are contending and requesting the shared wireless resources at what time can significantly alter the optimal allocation strategy. Considering both the starting and ending times of the uplink transmission as part of session characterization results in potentially $(S!)^2$ possible combinations.
Here, we only assume the starting time as the boundary of a session in the uplink phase to reduce the number of possible combinations to $S!$.
\begin{definition}[Session]
	\label{def: session}
	The boundary of a \textbf{session} is determined by the time instance where a UE completes receiving the downlink broadcasted model or 
    is ready to
    initiate its uplink transmission, i.e., local training is finished.
     \vspace{-.15cm}
\end{definition}
Similar definition also exists in~\cite{Vu_2022_sessionBasedMIMOFL_designs} in the context of MIMO.
A time-average RB allocation strategy is given during each session.
Due to the potentially large data size to compute and high BS transmit power, the downlink broadcasting is generally much shorter than the local computations. We therefore make the following realistic assumption also existing in~\cite{Xu_2024_TDMA_WFL, Vu_2022_sessionBasedMIMOFL_designs}:
\vspace{-.15cm}
\begin{assumption}
	\label{assumption: separate dl/ul}
	Uplink communication phase starts only after the end of all the downlink broadcasting communication.
	\vspace{-.15cm}
\end{assumption}
All system parameters will be redefined in the session-based formulation in the following.

\subsubsection{Downlink Sessions}
The broadcasting starts simultaneously for all UEs, with the successful receiving time depending on the UEs' channel states. Consequently, there is only one downlink session ordering: the channel strength ordering. We order UE indices in the descending order of channel strength:
\vspace{-.3cm}
\begin{equation}
	\gamma_1 \geq \gamma_2\geq \cdots \geq \gamma_S.
\end{equation}
The downlink session~$\ell' = 1,\ldots, S$ ends when the UE~$\ell'$ finishes the downlink transmission and start when $\ell'-1$-th UE finishes the downlink transmission, except the session~$1$ starts at time step zero. The duration of each downlink session is denoted~$t_{\ell'}^{(\mathrm{dl})}\geq0$. The $\ell'$-th UE completes receiving its downlink communication at: $
{\tau_{s,\mathrm{dl}}} = \sum_{s\leq \ell'} t_{s}^{(\mathrm{dl})}.$

\subsubsection{Uplink Sessions}
As for uplink sessions, each uplink session starts when the local computation/training of a UE is complete; it ends when either one UE completes its local training or all UEs complete the uplink communications, i.e., the CR ends.  
Denoting~$\sigma$ a permutation of~$\{1,\ldots, S\}$ representing the ordering of UE of completing their local computations. The uplink session~$\ell =1,\ldots, S$ starts when the UE~$\sigma(\ell)$ finishes its local training task, i.e., is ready for uplink transmission and ends when UE~$\sigma(\ell+1)$ (for~$\ell=1,\ldots,S-1$) is ready for uplink transmission or when the CR ends. The duration of session~$\ell$ is denoted $t_{\ell}^{(\mathrm{ul})}\geq 0$.

\subsubsection{Idle Time} 
Uplink sessions are assumed to start only after downlink sessions according to~\Cref{assumption: separate dl/ul}. In general, there can also exist an idle communication time between uplink and downlink phases denoted~$T_{\mathrm{idle}}\geq 0$. This idle time is missing in closely related work~\cite{Xu_2024_TDMA_WFL, Vu_2022_sessionBasedMIMOFL_designs} and is necessary to guarantee a good optimality of the scheduling design. 
\vspace{-.15cm}
\begin{remark}
	It is clear that the formulation with~$T_{\mathrm{idle}}$ achieves better results than without it, since the case without it is the special case of the current formulation with $T_{\mathrm{idle}}=0$. When the computation tasks of all UEs take significantly longer than transmission times, it becomes clear that optimal joint resource allocation and scheduling require $T_{\mathrm{idle}}>0$. Without this idle time, the slowest downlink UE would need to wait to complete its transmission simultaneously with the fastest UE initiating its uplink transmission.
\vspace{-.15cm}
\end{remark}

When coexisting with HB traffic, all~$K$ RBs are assigned to HB UEs during the FL idle time.

\subsubsection{Equivalent System Variables}

The \emph{delay} of the considered \ac{CR} becomes: $
T = \sum_{\ell'} t_{\ell'}^{(\mathrm{dl})}+ T_{\mathrm{idle}} + \sum_{\ell} t_{
	\ell}^{(\mathrm{ul})}.$
The\emph{ computational delay}~$\tau^{(\mathrm{cp})}_{s}$ that has a one-to-one relation with the computational frequency can be fully defined within the session definition framework. 
By definition, $T_{\mathrm{idle}}\geq 0$ represents the time between completion of the downlink broadcasting phase and the moment when the first UE~$\sigma(1)$ finishes its local training:
\begin{equation}
\tau_{\sigma(1)}^{(\mathrm{dl})} + \tau_{\sigma(1)}^{(\mathrm{cp})} = 	\tau_{S}^{(\mathrm{dl})}+ T_{\mathrm{idle}}= \sum_{\ell'} t_{\ell'}^{(\mathrm{dl})}+ T_{\mathrm{idle}}  .
	\label{eq: sigma 1 time constraint, T_idle}
	\vspace{-.07cm}
\end{equation}
By definition of duration of uplink sessions, the following satisfies for~$\ell=1,\ldots,S-1$: $
t_{\ell}^{(\mathrm{ul})}=  \tau_{\sigma(\ell+1)}^{(\mathrm{dl})}+\tau_{\sigma(\ell+1)}^{(\mathrm{cp})} - \tau_{\sigma(\ell)}^{(\mathrm{dl})}-\tau_{\sigma(\ell)}^{(\mathrm{cp})}.$
With telescopic sum, for all $s=1,\ldots,S-1$, the sum of the uplink session until the $s$ sessions is:
\begin{equation}
	\begin{aligned}
		\hspace{-.19cm}	\sum_{\ell\leq s} t_{\ell}^{(\mathrm{ul})}&= \tau_{\sigma(s+1)}^{(\mathrm{dl})}+\tau_{\sigma(s+1)}^{(\mathrm{cp})} - \tau_{\sigma(1)}^{(\mathrm{dl})}-\tau_{\sigma(1)}^{(\mathrm{cp})}\\[-.35cm]
		&=  \sum_{\ell'\leq \sigma(s+1)} t_{\ell'}^{(\mathrm{dl})}+\tau_{\sigma(s+1)}^{(\mathrm{cp})} -\sum_{\ell'} t_{\ell'}^{(\mathrm{dl})}- T_{\mathrm{idle}}.
	\end{aligned}
\end{equation}
We notice that~$s=0$ coincides with~\eqref{eq: sigma 1 time constraint, T_idle}.
The computation delay constraint can therefore be fully captured by a linear relation of session durations, for~$s=0,\ldots, S-1$:
\begin{equation}
	\tau_{\sigma(s+1)}^{(\mathrm{cp})} = \sum_{\ell\leq s} t_{\ell}^{(\mathrm{ul})}+\hspace{-.3cm}\sum_{\ell'>\sigma(s+1)}\hspace{-.3cm} t_{\ell'}^{(\mathrm{dl})} + T_{\mathrm{idle}}\geq \tau_{\sigma(s+1),\min}^{(\mathrm{cp})},
	\label{cons: computation time (session_pb)}
	\vspace{-.1cm}
\end{equation}
with $\tau_{s,\min}^{(\mathrm{cp})}$ the fastest computation time calculated with the maximal computational capacity~$f_{s,\max}$.

Denote other session-based variables: $K_{\ell'}^{(\mathrm{dl})}$ the number of RB allocated to the~$\ell'$ the downlink session, $K_{e, \ell'}^{(\mathrm{dl})}$ for HB traffic UE~$e$, the downlink communication of UE~$\ell'$ has to finish at the end of downlink session~$\ell'$:
\begin{equation}
(\forall \ell'\in\setS) \quad	\sum_{i\leq \ell'} r_{\ell'}^{(\mathrm{dl})}(K_{i}^{(\mathrm{dl})})t_{i}^{(\mathrm{dl})} \geq D.
	\label{cons: dl completion (session_pb)}
\end{equation}
The RB{ allocation constraint} during the downlink phase can be written as:
\vspace{-.2cm}
\begin{equation}
(\forall \ell'\in\setS) \quad\ K_{\ell'}^{(\mathrm{dl})} + \sum_eK_{e,\ell'}^{(\mathrm{dl})}\leq K.
	\label{cons: dl RB (session_pb)}
	\vspace{-.15cm}
\end{equation}

All UEs have to finish their transmission at the end of this \ac{CR}. UE~$\sigma(s)$ can only start its transmission at the $s$-th uplink session, therefore the following expression holds:
\begin{equation}
(\forall s\in\setS) \quad	\sum_{\ell\geq s}r^{(\mathrm{ul})}_{\sigma(s)}(K_{\sigma(s),\ell}^{(\mathrm{ul})}, p_{\sigma(s),\ell}^{(\mathrm{ul})})t_{\ell}^{(\mathrm{ul})} \geq D,
	\label{cons: ul completion (session_pb)}
\end{equation}
where $r^{(\mathrm{ul})}_{\sigma(s)}$ is the uplink rate function of UE~$\sigma(s)$ (uplink UE starting time ranking at~$s$th place)  at the uplink session~$\ell$, with~$\ell\geq s$; $K_{\sigma(s),\ell}^{(\mathrm{ul})}$ and $p_{\sigma(s),\ell}^{(\mathrm{ul})}$ the RB and the power allocated to UE~$\sigma(s)$ at the uplink session~$\ell$.
The RBs again are shared with HB traffic UEs, at each session~$\ell$,
\begin{equation}
(\forall \ell\in\setS) \quad	\sum_{s\leq \ell}K_{\sigma(s),\ell}^{(\mathrm{ul})} + \sum_{e\in\setE}K_{e,\ell}^{(\mathrm{ul})}\leq K,
	\label{cons: ul RB (session_pb)}
	\vspace{-.15cm}
\end{equation}
with~$K_{e,\ell}^{(\mathrm{ul})}$ the number of RB allocated to~$e$-th HB traffic UE.
With the introduced session-based variables that replace~$\tau_s^{(\mathrm{cp})}$, the total consumed energy for FL by UE~$\sigma(s)$ with~$s=1,\ldots,S$ can be written as:
\begin{equation}
	\begin{aligned}
		E_{\sigma\!(\!s\!)}^{\!(\!\mathrm{tot}\!)\!} \!\!=\! \sum_{\ell\geq s }p_{\sigma(s),\ell}^{(\mathrm{ul})}t_{\ell} \!+\!  \frac{\kappa\zeta^3\Theta_{\sigma(s)}^3}{(\hspace{-.12cm}\sum\limits_{\ell\leq s-1} \hspace{-.22cm}t_{\ell}^{(\mathrm{ul})} \!\!+ \hspace{-.32cm}\sum\limits_{\ell'>\sigma(s)}\hspace{-.3cm} t_{\ell'}^{(\mathrm{dl})} \!\!+\! T_{\mathrm{idle}} )^2}\!\leq\! E_{\sigma\!(\!s\!)\!,\mathrm{budget}}.
	\label{cons: energy_cstr (session_pb)}\\[-.4cm]
	\end{aligned}
\end{equation}

For any~$e\in\setE$, the average HB traffic rate over the whole process is equal to:
\begin{equation}
	\hspace{-.2cm}
		\frac{1}{T} \!\sum_{t=1}^Tr_e^{(t)}\!  =\! \frac{\sum_{\ell'} r_{e,\ell'}^{(\mathrm{dl})}t_{\ell'}^{(\mathrm{dl})} \!+\! T_{\mathrm{idle}} {r_{e,\mathrm{idle}}} \!+\!\sum_{\ell} r_{e,\ell}^{(\mathrm{ul})}t_{\ell}^{(\mathrm{ul})} }{ \sum_{\ell'} t_{\ell'}^{(\mathrm{dl})} + T_{\mathrm{idle}}+ \sum_{\ell} t_{\ell}^{(\mathrm{ul})} },
		\vspace{.05cm}
\end{equation}
{where~$r_{e,\mathrm{idle}}$ is the rate of UE~$e$ during the FL communication idle time.
}
The HB traffic requirement~\eqref{eq: time_raw_eMBB} can be therefore written as:
\begin{equation}
\begin{aligned}
	(\forall e\in\setE)\quad \sum_{\ell'} r_{e,\ell'}^{(\mathrm{dl})}t_{\ell'}^{(\mathrm{dl})} + T_{\mathrm{idle}} {r_{e,\mathrm{idle}}} +\sum_{\ell} r_{e,\ell}^{(\mathrm{ul})}t_{\ell}^{(\mathrm{ul})} \\[-.18cm]
\geq \theta \Big( \sum_{\ell'} t_{\ell'}^{(\mathrm{dl})}+ T_{\mathrm{idle}} + \sum_{\ell} t_{\ell}^{(\mathrm{ul})} \Big).
\label{cons: eMBB each e (session_pb)}
\end{aligned}
\end{equation}

		\vspace{-.37cm}
	\subsection{Session-based Problem Formulation}
	\vspace{-.1cm}
	Combining all the constraints and the reformulation considerations, given uplink order~$\sigma$, the problem can be written as follows:
	\begin{subequations}
		\label{pb: session-based problem}
		\begin{align}
		(\setP_{\sigma}):	\!\min\  & \sum_{\ell'} t_{\ell'}^{(\mathrm{dl})}+ T_{\mathrm{idle}}+ \sum_{\ell} t_{\ell}^{(\mathrm{ul})} ,
			\label{obj: session-based problem}
			\\[-.1cm]
			\mathrm{s.t.}\ \  & \eqref{cons: computation time (session_pb)}, \eqref{cons: dl completion (session_pb)}, \eqref{cons: dl RB (session_pb)}, \eqref{cons: ul completion (session_pb)}, \eqref{cons: ul RB (session_pb)},\eqref{cons: energy_cstr (session_pb)}, \eqref{cons: eMBB each e (session_pb)},  \notag
			\\[-.1cm]
		& \hspace{-1.8cm}(\forall s\in\setS)(\forall\ell\in\setS)\    p_{\sigma(s), \ell}^{(\mathrm{ul})}\in[0,P_{\max}] 
			\label{cons: power (session_pb)}
			\\[-.1cm]
	& 
	 \hspace{-1.8cm}(\forall s,\! \ell,\! \ell'\!\in\!\setS)(\forall e\!\in\!\setE)\   t_{\ell'}^{(\!dl\!)}\!\!,\! t_{\ell}^{(\!ul\!)}\!,\! T_{\mathrm{idle}},\! K_{s,\ell}^{(\!ul\!)},\! K_{\ell'}^{(\!dl\!)},\! K_{e,\ell'}^{(\!dl\!)},\!K_{e,\ell}^{(\!ul\!)}   \!\!\geq \!0.
			\label{cons: positivity (session_pb)}\\[-.8cm]
			\notag
		\end{align}
	\end{subequations}
	Denoting the feasible set of problem variables as:  $$\setX \!=\! \{(t_{\ell'}^{(\mathrm{dl})}\!, t_{\ell}^{(\mathrm{ul})}\!, T_{\mathrm{idle}}, K_{s,\ell}^{(\mathrm{ul})}\!, K_{\ell'}^{(\mathrm{dl})}\!, K_{e,\ell'}^{(\mathrm{dl})}\!,K_{e,\ell}^{(\mathrm{ul})}\!, p_{s, \ell}^{(\mathrm{ul})})_{e,s,\ell,\ell'}\!\}.$$
	
	The same feasibility condition as $(\setP_{\mathrm{rig}})$ holds. 
	\vspace{-.18cm}
	\begin{theorem}[Feasibility Condition]
		The problem~$(\setP_{\sigma})$ is feasible if and only if $D<E_{s,\mathrm{budget}}B\gamma_s$ and $K>a\theta$. 
		\label{theorem: feasibility (session_pb)}
		\vspace{-.55cm}
	\end{theorem}
    \begin{proof}
    In Appendix~\ref{appendix: proof feasibility session}.
    \vspace{-.3cm}
    \end{proof}
    Although the conditions for feasibility are the same, the achieved latency is inherently at least as good as what rigid allocation achieves. The significant gap will be shown in the simulations. We start by analyzing certain properties of the formulated session-based problem.

		\vspace{-.35cm}
	\subsection{Discussion on Scheduling Properties}
	\vspace{-.1cm}
	\label{sec: scheduling}
	
	The resulting optimization problem is a \emph{JCSRA problem} within one FL CR, as described in the introduction. 
	Besides the numerous \emph{resource allocation} FL frameworks, \emph{CS} has rarely been studied~\cite{Li_FLCommScheduling_2021, Xu_2024_TDMA_WFL}. 
	The major difference of the proposed multi-server JCSRA with them comes from the fact that the uplink transmission rate~\eqref{eq: uplink rate (time raw)} is non-linear w.r.t. the resources given.
    In the linear rate case, the optimal solution of JCSRA problem has a single-server system structure, i.e., only one UE can perform the uplink communication at each time~\cite{Li_FLCommScheduling_2021, Xu_2024_TDMA_WFL}. 
     In realistic scenario with the average rate expression~\eqref{eq: uplink rate (time raw)} derived from~\cite{Shen_OFDM_allocation_2005}, the JCSRA problem in general does not consist of a single-server system.
     \vspace{-.15cm}
	\begin{proposition}
		\label{th: multi-queue_uplink}
		The optimal communication scheduling for the uplink communication, in general, does not consist of a single-server system.
		\vspace{-.15cm}
	\end{proposition}
	\begin{proof}
        First, the multi-server based formulation is at least as good than single-server based formulation, since the formulation~$(\setP_{\sigma})$ includes single-server system solution, by adding constraints $p_{\sigma(s), \ell}=0$ for $\ell>s$ for all~$s$.
        The optimal solution of a multi-server system is not always a single-server solution.        
		This comes from the non-linearity of the rate expression. One UE transmitting over multiple RBs simultaneously results in lower SNR at each RB. When the `unit' SNR is low, it is more advantageous to assign some subchannels to other UEs. Counterexamples are easy to find.
		\vspace{-.25cm}
	\end{proof}
	
	Li et al.~\cite{Li_FLCommScheduling_2021} established the non-preemptive and non-idle properties of the optimal solution in the single-server case. We establish similar properties under multi-server JCSRA scenarios, where the proof is less trivial.
\vspace{-.15cm}
	\begin{theorem}
		\label{theorem: non-preemptive_non-idle}
		There exists an optimal communication session scheduling that is \textbf{non-preemptive} and \textbf{non-idle} within each downlink and uplink phase.  
		\vspace{-.15cm}
	\end{theorem}
    \begin{proof}
        In Appendix~\ref{appendix: proof non preemptive non idle}.
        \vspace{-.35cm}
    \end{proof}

		\vspace{-.15cm}
	\subsection{Algorithm Development}
	\vspace{-.08cm}
	Xu~\cite{Xu_2024_TDMA_WFL} considers a single-server JCSRA problem, establishes closed-form solutions for certain variables and transforms the problem to a convex optimization problem. However, similar approaches are not applicable in multi-server scenarios, where variable coupling is significantly more complex. Therefore, we develop a specific algorithm tailored for this problem.
	We now present the algorithm for solving the problem~$(\setP_{\sigma})$, which considers in addition the coexistence with HB traffic UEs. First, we reformulate the problem to reduce its dimensionality, after which we develop an algorithm that addresses the non-convexity and non-separable constraints.
	
	\subsubsection{HB Traffic Constraint Reformulation}

	The only interaction that HB \ac{UE}s have with FL \ac{UE}s is through the number of RBs shared during each phase. We introduce a slack variable for an arbitrary downlink phase~$\ell'$ (reps. uplink phase~$\ell$), defined as the sum of RBs occupied~$K_{\mathrm{HB}, \ell'}^{(\mathrm{dl})}$ (resp. $K_{\mathrm{HB}, \ell}^{(\mathrm{ul})}$, such that:
	\begin{equation}
	(\forall\ell'\in\setS)\quad	\sum_{e} K_{e,\ell'}^{(\mathrm{dl})} \leq K_{\mathrm{HB}, \ell'}^{(\mathrm{dl})},
	\end{equation}
	and 
	\begin{equation}
	(\forall\ell\in\setS)\quad	\sum_{e} K_{e,\ell}^{(\mathrm{ul})} \leq K_{\mathrm{HB}, \ell}^{(\mathrm{ul})}.
	\vspace{-.15cm}
	\end{equation}
	\vspace{-.15cm}
	\begin{proposition}
    \label{theorem: HB reformulation}
		The problem~$(\setP_{\sigma})$ is equivalent to the problem~$(\setP 1_{\sigma})$ by replacing the variables $K_{e,\ell'}^{(\mathrm{dl})}$ and $K_{e,\ell}^{(\mathrm{ul})}$ for~$e\in\setE$ by $K_{\mathrm{HB},\ell'}^{(\mathrm{dl})}$ and $K_{\mathrm{HB},\ell}^{(\mathrm{ul})}$ respectively and the constraint~$\eqref{cons: eMBB each e (session_pb)}$ by 
		\begin{multline}
			 \sum_{\ell'} K_{\mathrm{HB},\ell'}^{(\mathrm{dl})}t_{\ell'}^{(\mathrm{dl})} +\sum_{\ell} K_{\mathrm{HB},\ell}^{(\mathrm{ul})}t_{\ell}^{(\mathrm{ul})} + T_{\mathrm{idle}}K
			\\ 
			\geq a \theta \Big( \sum_{\ell'} t_{\ell'}^{(\mathrm{dl})}+T_{\mathrm{idle}}+ \sum_{\ell} t_{\ell}^{(\mathrm{ul})} \Big).
			\label{cons: embb transformed (session_pb)}
		\end{multline}
		\vspace{-.15cm}
	\end{proposition}
    \begin{proof}
        In Appendix~\ref{appendix: proof HB reformulation}.
        \vspace{-.15cm}
    \end{proof}
    By the proposition, only the total RBs demanded by HB traffic need to be considered in~$(\setP 1_{\sigma})$. The resulting optimization problem becomes, therefore, scalable to the number of HB traffic UEs in the network.

	\subsubsection{Non-Convexity Handling}
	The problem is non-convex due to the product term in the communication energy constraint~\eqref{cons: energy_cstr (session_pb)}, in the product between the time and rate in the transmission completion constraints of downlink~\eqref{cons: dl completion (session_pb)}, uplink~\eqref{cons: ul completion (session_pb)}, HB~\eqref{cons: embb transformed (session_pb)}. The computing energy is to be constrained from above, while the completion constraints must ensure a minimum, thus the optimization will proceed in the directions of minimization and maximization, respectively, which require distinct handling.

	\paragraph{Maximizing Product}

	The product consists of the product of the durations with either RB allocation variables as in~\eqref{cons: dl completion (session_pb)}, \eqref{cons: embb transformed (session_pb)}, or with the concave (see Appendix~\ref{appendix: proof convex problem}) uplink rate expression w.r.t. power and RB allocation.
	Each product can be seen as the quotient of a concave function and a convex function~$1/t$ with $t$ any duration variable. 
	With the development in fractional programming, we employ the well-known \emph{quadratic transform}~\cite{Shen_FP_quadraticTransform_2018} to handle the product terms in order to obtain a stationary point with an iterative algorithm.

	Using quadratic transform for handling product terms, given a concave function~$X:x\mapsto X(x)$ to multiply with a certain duration variable~$t$, the product term can be transformed as: $(\forall (x,t)\in\setX\times \R^+)$
	\vspace{-.1cm}
	\begin{equation} 
		\hspace{-.2cm}g(x,t) \defeq X(x)t \\ 
= \max_y\! \Big(\!2y\sqrt{X(x)} - \!\frac{y^2}{t}\!\Big)\! \defeq\! {\max_y}\,  \hat{g}(x,\!t,\!y).
	\label{eq: quadratic transform}
	\end{equation}
	The variable~$y$ is introduced as an auxiliary variable. The transform has the advantages of:
	\begin{itemize}
		\item Equivalent solutions: $(x^*, t^*)$ maximizes of~$g$ if and only if $(x^*, t^*, y^*)$ maximizes~$ \hat{g}$ for chosen~$y^*$,
		\item Equivalent objective: as already stated in~\eqref{eq: quadratic transform}, for any~$(x,t)$, the equality holds with~$g(x,t)=\hat{g}(x,t, y^*)$ with~$y^*=\arg\min \hat{g}(x,t,y) = \sqrt{X(x)}t$.
	\end{itemize}
	
	Using the transform, we introduce a slack variable for each product term in each non-convex constraint as follows, while denoting the resulting constraint~$(x)$ with the notation~$\widehat{(x)}$: 
	\begin{itemize} 
		\item $(\widehat{\ref{cons: embb transformed (session_pb)}})$:  $y^{(\mathrm{dl})}_{\mathrm{HB}, \ell'}$ and $y^{(\mathrm{ul})}_{\mathrm{HB},\ell}$ in~\eqref{cons: embb transformed (session_pb)}:
		\begin{multline}
\hspace{-.9cm}	
\sum_{\ell'} \Bigg(\!2y^{(\mathrm{dl})}_{\mathrm{HB}, \ell'}\sqrt{K_{\mathrm{HB},\ell'}^{(\mathrm{dl})}} -\frac{y^{(dl)2}_{\mathrm{HB}, \ell'}}{ t_{\ell'}^{(\mathrm{dl})}}\!\Bigg) +\sum_{\ell} \Bigg(\!2y^{(\mathrm{ul})}_{\mathrm{HB}, \ell}\sqrt{K_{\mathrm{HB},\ell}^{(\mathrm{ul})}}
 \\
\hspace{-.7cm}
-\!\frac{y^{(ul)2}_{\mathrm{HB}, \ell}}{ t_{\ell}^{(\mathrm{ul})}}\!\Bigg) 
			\!\geq \! a\theta \Big(\! \sum_{\ell'} t_{\ell'}^{(\mathrm{dl})}\!\! +\! \sum_{\ell} t_{\ell}^{(\mathrm{ul})}\!\!+ T_{\mathrm{idle}}\Big) - KT_{\mathrm{idle}}, \tag{$\widehat{\ref{cons: embb transformed (session_pb)}}$}
		\end{multline}
		\item $(\widehat{\ref{cons: dl completion (session_pb)}})$:  $y^{(\mathrm{dl})}_{\ell'}$ in~\eqref{cons: dl completion (session_pb)}:
		\begin{equation}
		\hspace{-0.8cm}
        (\forall \ell'\in\setS) \sum_{i\leq \ell'} \Bigg( 2y^{(\mathrm{dl})}_{\ell'}\sqrt{r_{\ell'}^{(\mathrm{dl})}(K_{i}^{(\mathrm{dl})})} - \frac{y^{(dl)2}_{\ell'}}{t_{i}^{(\mathrm{dl})}} \Bigg)\geq D,
		\tag{$\widehat{\ref{cons: dl completion (session_pb)}}$}
		\end{equation}
		\item $(\widehat{\ref{cons: ul completion (session_pb)}})$: $y^{(\mathrm{ul})}_{s,\ell}$ in~\eqref{cons: ul completion (session_pb)}: $	(\forall s\in\setS) $
		\begin{equation}
%		\hspace{-1cm}
		\sum_{\ell\geq s}\!\!\Bigg(\! 2y^{(\mathrm{ul})}_{s,\ell} \! \sqrt{\! r^{(\mathrm{ul})}_{\sigma(s)}( K_{\sigma(s),\ell}^{(\mathrm{ul})}, p_{\sigma(s),\ell}^{(\mathrm{ul})})} - \frac{y^{(ul)2}_{s,\ell}}{t_{\ell}^{(\mathrm{ul})}}\! \Bigg) \!\! \geq\! D.
		\tag{$\widehat{\ref{cons: ul completion (session_pb)}}$}
		\end{equation}
	\end{itemize}

	The slack variables updates will be detailed together in the section~\Cref{sec: algorithm}.
	\paragraph{Minimizing Product Term}
	The product term of the communication power and duration of communication in the transmit energy is to be upper bounded. We aim to find a tight convex approximation. Using the principle of \ac{MM}, for any point $(\hat{p}_{\sigma(s),\ell}^{(\mathrm{ul})},\hat{t}_{\ell}^{(\mathrm{ul})})$, a tight convex upper bound can be found, for all~$s,\ell\in\setS$:
	\begin{equation}
		\hspace{-.17cm}p_{\sigma(s),\ell}^{(\mathrm{ul})}t_{\ell}^{(\mathrm{ul})}\leq \frac{p_{\sigma(s),\ell}^{(ul)2}}{2y^{(E)}_{s,\ell}} + \frac{y^{(E)}_{s,\ell}t_{\ell}^{(ul)2}}{2} \defeq {\hat{\phi}_{s,\ell}}(p_{\sigma(s),\ell}^{(\mathrm{ul})}, t_{\ell}^{(\mathrm{ul})}),
		\label{eq: MM surrogate def}
	\end{equation}
	with~$y_{s,\ell}^{(E)} =\hat{p}_{\sigma(s),\ell}^{(\mathrm{ul})} / \hat{t}_{\ell}^{(\mathrm{ul})} $. The function ${\hat{\phi}_{s,\ell}}$ is convex and the inequality is tight at the point $(\hat{p}_{\sigma(s),\ell}^{(\mathrm{ul})},\hat{t}_{\ell}^{(\mathrm{ul})})$. Denote the approximated total energy $\hat{E}^{(\mathrm{tot})}_{\sigma(s)} = E^{(\mathrm{cp})}_{\sigma(s)} + \sum_{\ell}{\hat{\phi}_{s,\ell}}$. The constraint~\eqref{cons: energy_cstr (session_pb)} is transformed as such to~$(\widehat{\ref{cons: energy_cstr (session_pb)}})$.

	\subsubsection{Algorithm}
	\label{sec: algorithm}
	Given any feasible point of the problem $X\in\setX$, the updates conducted on the auxiliary variables in set~$\setY=\{(y_{\mathrm{HB},\ell'}^{(\mathrm{dl})}, y_{\mathrm{HB},\ell}^{(\mathrm{ul})}, y^{(\mathrm{dl})}_{\ell'}, y_{s, \ell}^{(\mathrm{ul})}, y_{s,\ell}^{(E)})\}\subset \R_{++}^{3S+S(S+1)}$ are:
	\begin{equation}
\begin{cases}
	\!\!\!(\forall \ell'\in\setS) &\hspace{-.3cm} y_{\mathrm{HB},\ell'}^{(\mathrm{dl})} \!\!=\!\! \sqrt{\!K_{\mathrm{HB},\ell'}^{(\mathrm{dl})}}t_{\ell'}^{(\mathrm{dl})}, y^{(\mathrm{dl})}_{\ell'}\!\! =\!\! \sqrt{\!K_{\ell'}^{(\mathrm{dl})}}\!t_{\ell'}^{(\mathrm{dl})}\\
\!\!\!(\forall \ell\in\setS)\ &\hspace{-.3cm}  y_{\mathrm{HB},\ell}^{(\mathrm{ul})} = \sqrt{K_{\mathrm{HB},\ell}^{(\mathrm{ul})}}t_{\ell}^{(\mathrm{ul})},\\
\!\!\!(\forall s\in\setS)(\forall \ell\in\setS) &\hspace{-.3cm}  y_{s, \ell}^{(\mathrm{ul})} = \sqrt{r_{\sigma(s), \ell}^{(\mathrm{ul})}}t_{\ell}^{(\mathrm{ul})}, y_{s,\ell}^{(E)}=\frac{p_{\sigma(s),\ell}^{(\mathrm{ul})}}{t_{\ell}^{(\mathrm{ul})}}.\\[-.35cm]
\end{cases}
		\label{eq: update on y}
	\end{equation}

	The resulting transformed problem~$(\setT\setP_{\sigma})$ from~$(\setP 1_{\sigma})$ given $Y\in\setY$ is:
	\begin{subequations}
		\label{pb: session-based problem: convex iterative subproblem}
		\begin{align}
			\hspace{-0cm}(\setT\setP_{\sigma})\!\!:	\!\min_{X\in\setX}  & \sum_{\ell'} t_{\ell'}^{(\mathrm{dl})}\!\!+\! T_{\mathrm{idle}}\!+\! \sum_{\ell} t_{\ell}^{(\mathrm{ul})}
			\\[-.1cm]
		\hspace{-.9cm}	\mathrm{s.t.}\ \  &  \eqref{cons: computation time (session_pb)}, (\widehat{\ref{cons: dl completion (session_pb)}}), \eqref{cons: dl RB (session_pb)},  (\widehat{\ref{cons: ul completion (session_pb)}}), 
			 \eqref{cons: ul RB (session_pb)},(\widehat{\ref{cons: energy_cstr (session_pb)}}),(\widehat{\ref{cons: embb transformed (session_pb)}}),  
             \eqref{cons: power (session_pb)},\notag \\[-.1cm]
			& \hspace{-1.69cm}(\forall s,\! \ell,\! \ell'\!\in\!\setS)\ t_{\ell'}^{(\mathrm{dl})}\!\!, t_{\ell}^{(\mathrm{ul})}\!\!, T_{\mathrm{idle}}, K_{s,\ell}^{(\mathrm{ul})}\!, K_{\ell'}^{(\mathrm{dl})}\!\!, K_{\mathrm{HB},\ell'}^{(\mathrm{dl})},K_{\mathrm{HB},\ell}^{(\mathrm{ul})}   \geq 0.
			\label{cons: positivity (session_pb)}
		\end{align}
	\end{subequations}
	\begin{theorem}
		Given fixed auxiliary variables $Y\in\setY$, the subproblem~$(\setT\setP_{\sigma})$ is a convex optimization problem.
		\label{theorem: convex TP}
	\end{theorem}
    \begin{proof}
        In Appendix~\ref{appendix: proof convex problem}.
    \end{proof}
	The proposed algorithm to solve the optimization problem~\eqref{pb: session-based problem} given an arbitrary ordering~$\sigma$ is detailed in~\Cref{algo: iterative giving ordering}. It has a guarantee to converge to a stationary point.
	\begin{algorithm}\small
		\SetAlgoLined
		\textbf{Initialize:} $X_0\in\setX$, $T_0 = \infty$, $\varepsilon=10^{-4}$, $n_{\max} = 100$.  \\
		
		\For{$n = 1,\ldots$}{
			\underline{\emph{Update of the auxiliary variables~$Y$}}\\
			Compute $Y_{n}$ according to~\eqref{eq: update on y} based on~$X_{n-1}$\;
			\underline{\emph{Update of the original variables~$X$}}\\
			Solve the convex optimization problem~$(\setT\setP_{\sigma})$ given~$Y_{n}$:  compute~$X_n$ with achieved optimum~$T_n$\;
			\underline{\emph{Stopping criterion}}\\
			\If{$\|T_n-T_{n-1}||/\|T_n\| \leq \varepsilon$ or $n \geq n_{\max}$}{Stop loop}
		}
		\KwResult{$X_n$ and~$T_n$.}
		\caption{Iterative algorithm solving $(\setP_{\sigma})$~\eqref{pb: session-based problem}}
		\label{algo: iterative giving ordering}
	\end{algorithm}
	\vspace{-.45cm}
	
	\begin{theorem}
		The sequence~$(T_n)_{n\in\N}$ of~\Cref{algo: iterative giving ordering} is a monotonically decreasing sequence and~$(X_n, Y_n)_{n\in\N}$ converges to a stationary point of~$(\setP_{\sigma})$.
		\label{theorem: algo convergence}
			\vspace{-.15cm}
	\end{theorem}
    \begin{proof}
    	{In Appendix~\ref{appendix: proof algo convergence}.}
        \vspace{-.35cm}
    \end{proof}
	
	\vspace{-.25cm}
	\subsection{Heuristic Ordering: Rigid-based Ordering}
	\vspace{-.1cm}
	\label{sec: heuristic Ordering}
	The developed multi-server JCSRA algorithm is for a given ordering~$\sigma$ for the uplink session starting time.
	Evaluating every possible combination of orderings is clearly NP-hard.
    Therefore, we utilize the results from the rigid resource allocation discussed in~\Cref{sec: rigid} as a heuristic for determining this ordering.
	The intuition is that the optimal solution derived from the rigid resource allocation within different constraints provides a reliable indication of the appropriate ordering strategy.
	The heuristic has been confirmed via simulations in~\Cref{sec: simul_verify}.

	\vspace{-.45cm}
	\section{Simulations}
	\vspace{-.12cm}
	\subsection{Simulation Settings}
	\vspace{-.1cm}
	\subsubsection{Settings}
	It is considered that at one arbitrary CR, $S=10$ FL UEs are participating in the training in a resource-constrained wireless system with~$K=10$ RBs coexisting with~$20$ HB traffic UEs. All UEs, FL and HB traffic UEs are uniformly distributed in the cell of radius of $\SI{50}{\meter}$. 
	The system has an SCS of \SI{60}{\kilo\hertz} where each RB contains 12 subcarriers.
	The ``long-term" average of the stationary channel is only subject to free-space path loss. 
	The HB traffic minimum rate requirement among all HB traffic UEs is of $\SI[per-mode = symbol]{10}{\mega\bit\per\second}$. We assume that FL UEs train a neural network of model parameter size of~$D=\SI{800}{\mega\bit}$ on a dataset of the same size and dimension of Cifar-10~\cite{cifar}, i.e., in total of 60000 RGB images of size $32\times 32$ with floating points in $32$-bits, distributed among UEs. The 60000 images are assumed to be distributed among the~$S$ UEs according to random ratios in order to simulate imbalanced quantity of local data, hence system heterogeneity. The computation cycle required for one sample $C_s$ is calculated as $15$ cycles per bit~\cite{Do_2022_DRLUAVFL_FLcomputEnergyRef} multiplied by the number of bits contained in one data sample. FL UEs perform local SGD updates of 20 epochs (to reduce the overall CR~\cite{mcmahan_FL_2017}). 
	The complete system parameters are detailed in the~\Cref{tab: param}. For ease of comparison, the UE-wise energy budget constraint is transformed to a network-wide sum energy constraint of~$E_{\mathrm{budget}}$.
	
	\input{tables/system_params_j_final}
	\subsubsection{Baselines}

	The baseline methods to compare with are listed as follows:
	\begin{itemize}
		\item Time-uniform \emph{rigid} RB allocation: detailed in~\Cref{sec: rigid}, which most existing RB allocation work on FL is based on. 
		\item Consider FL as an HB traffic: max-sum-rate (\emph{MSR}) and max-min-rate (\emph{MMR}), to show the importance of having a dedicated service class than HB traffic. Note the exact energy planning in this case is not possible, all UEs perform local training with their maximum speeds.
        \item {Single-server JCSRA (sequential transmission)}~\cite{Xu_2024_TDMA_WFL}: only one UE can be assigned for communication at each uplink session~$\ell$ as assumed in all existing CS literature~\cite{Xu_2024_TDMA_WFL, Li_FLCommScheduling_2021}. 
		\end{itemize}
		
		\vspace{-.4cm}
		\subsection{Algorithm Convergence + Effect of Ranking}
		\vspace{-.1cm}
		\label{sec: simul_verify}
	
		The proposed solution consists of an iterative algorithm. Its convergence and performance are verified in~\Cref{fig: convergence}. We initialize the proposed algorithm with the rigid-based formulation results. The algorithm produces feasible solutions at each iteration (energy level, remains stable at about \SI{200}{\joule}), and improves the latency by about \SI{38}{\percent} in 29 iterations. This confirms the monotonicity of the proposed algorithm and the significant gap that rigid-based formulation introduces to the JCSRA approach.
        Compared to common HB power allocation strategy (MSR and MMR), for which no energy planning is possible, MMR achieves higher latency with even much higher energy, while MSR does provide \SI{2.7}{\percent} less latency, while using 4.8$\times$ more energy than available.

		\begin{figure*}
			\begin{minipage}[t]{0.325\textwidth}
				\centering
				\includegraphics[width=\linewidth]{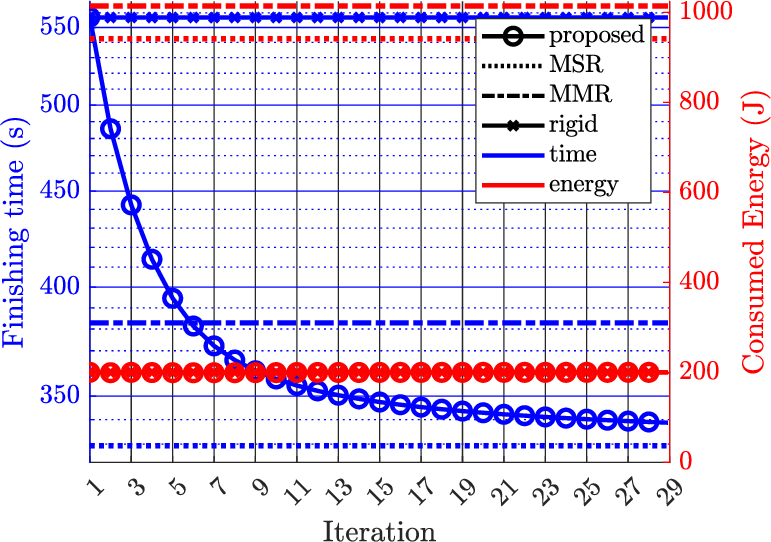}
				\caption{Algorithm convergence with $E_{\mathrm{budget}}$ of value \SI{200}{\joule}. {Legend: Line style indicates method (black lines); color indicates time/energy.}}
				\label{fig: convergence}
			\end{minipage}
			\begin{minipage}[t]{0.30\textwidth}
				\centering
				\includegraphics[width=\linewidth]{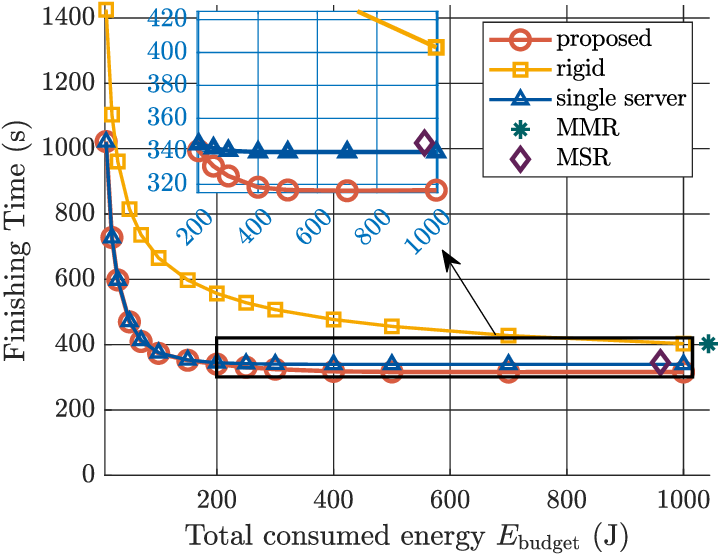}
				\caption{Energy-Time Pareto front}
				\label{fig: pareto}
			\end{minipage}
			\begin{minipage}[t]{0.325\textwidth}
				\centering
				\vspace{-4.15cm}
				\includegraphics[width=0.5\linewidth, trim=0 -5 0 0 ]{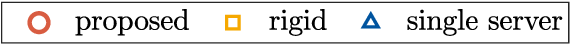}
				\includegraphics[width=0.9\linewidth]{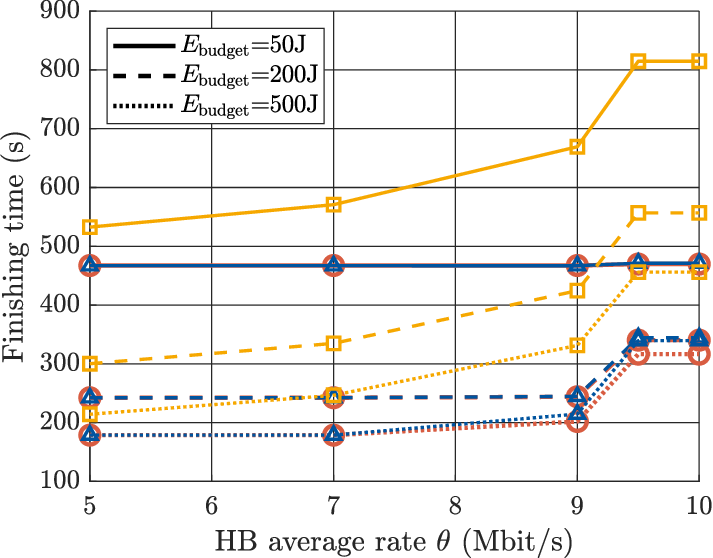}
				\vspace{-.25cm}
				\caption{$\theta$ influence. {Legend: Line style indicates $E_{\mathrm{budget}}$ (black lines); color/marker indicates methods.}}
				\label{fig: theta influence}
			\end{minipage}
			\vspace{-.6cm}
		\end{figure*}

		As specified in~\Cref{sec: heuristic Ordering}, the proposed method uses the ranking given by the rigid formulation. Its performance{ is compared against} the {optimal }ranking, obtained by iterating over all combinations {(5070 ranking combinations with~$S=7$ FL clients considered} here due to{ limited computation capacity)} w.r.t. different system parameters, {as }shown in the comparison \Cref{tab: ranking_confirmation}. {For evaluations w.r.t.~$E_{\mathrm{budget}}$ and the coexisting HB traffic requirement~$\theta$, the }UE's location is generated randomly within a disk{ and the data quantity's split ratio among clients is also randomly generated}. Each random realization characterizes a potentially different system nature{ in terms of network condition and data quantity distributions}.{ We further isolate the effect of data amount heterogeneity by assigning equal channel conditions and splitting data among UEs based on the Dirichlet distribution~\cite{hsu_dirichlet_2019} with parameter~$\alpha$ to characterize the data amount heterogeneity. A smaller value of~$\alpha$ indicates greater heterogeneity within the system{, so the data quantity is highly imbalanced among UEs}.} We evaluate the{ optimal} latency achieved {given }the rigid-based ranking with{ the best performing one among} all possible ranking combinations. 
		{For each parameter setting, 10 random realizations were conducted.
		The percentage of random realizations with relative gaps (ratio of the difference of the achieved latency with the best performing latency) below a threshold is evaluated. We observe that for all evaluated system parameters, all evaluated realizations have their relative gap below $1\%$, nearly all realizations below $0.5\%$, and a big majority were even below~$0.1\%$.}
		This confirms heuristically the{ robustness of the} rigid-based ranking's performance under{ different network conditions and data amount distribution}.
		\input{tables/res_heuristic_mr_final}
		
		\vspace{-.3cm}
		\subsection{Performance Evaluation}
		\vspace{-.08cm}
		The effect of different system parameters on the achieved latency will be shown in this section. The energy budget/finishing time Pareto front is shown in~\Cref{fig: pareto}. It can be observed that the proposed method and single-server have non-negligible gain compared to \emph{rigid} formulation for all energy values. The gain is especially large in the moderate energy constraint region 50-\SI{150}{\joule}, for instance at \SI{50}{\joule} the gain is about \SI{40}{\percent}. 
        The MMR and MSR both consume significant energy. MSR coincides with a single server in a high-energy consumption region.

		The effect with HB-traffic requirement has been shown in~\Cref{fig: theta influence}. Rigid-based formulation is more sensitive to the increase of the HB-traffic increase, while the proposed method and single-server remains stable for a certain region of $\theta$ increase and later increase. In the case of \SI{50}{\joule}, the proposed method remains uniform for different values of~$\theta$, resulting in a high gap between rigid and JCSRA for high HB traffic requirement.

        The effect w.r.t. other parameters such as the number of total RBs available and the model parameter size is shown in \Cref{fig: params eval}. The superiority of JCSRA (proposed method and single-server) with rigid-based formulation is consistent. Time-dependent optimization (JCSRA) is less sensitive to the increasingly stringent system  constraint~$K$, and increasing transmission tasks to complete~$D$. 
		\begin{figure}
			\centering
			\includegraphics[width=0.45\linewidth, trim=0 20 0 0]{final_figures_eps/j6_common_legend.eps}
			\vfill
			\subfloat[{$K$ influence}\label{fig: Ecstr K}]{\includegraphics[width=0.48\linewidth, trim=0 0 0 0]{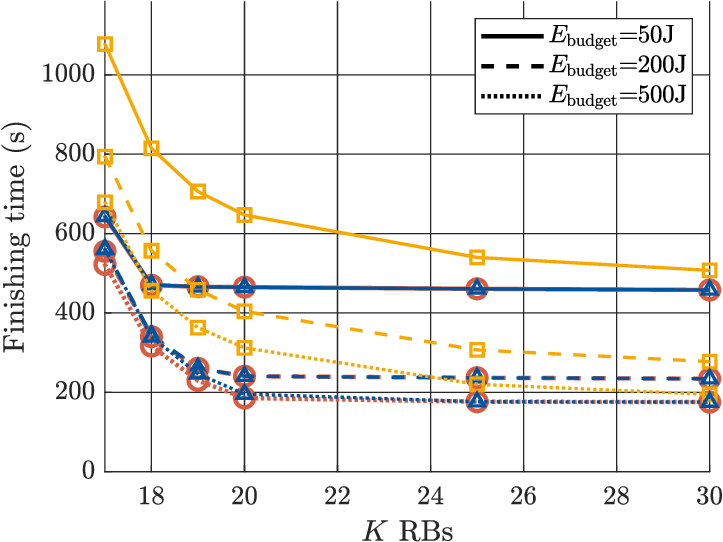}
			}
			\hfill
			\subfloat[ $D$ influence\label{fig: Ecstr theta}]{\includegraphics[width=0.48\linewidth, trim=0 0 0 0]{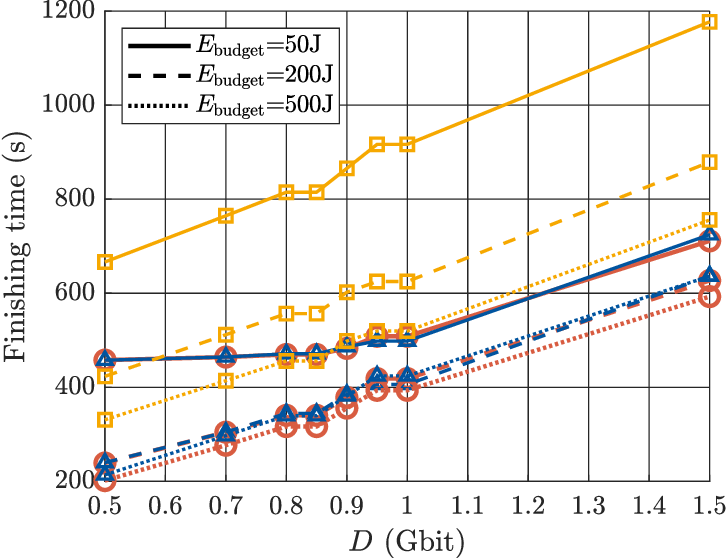}
			}
			\caption{Parameter influence when fixed energy constraint is imposed.  {Legend: Line style indicates $E_{\mathrm{budget}}$ (black lines); color/marker indicates methods.}}
			\label{fig: params eval}
			\vspace{-.45cm}
		\end{figure}

    \vspace{-.3cm}
        \subsection{Proposed and Single-Server Gap Evaluation}
        \label{sec: simulation gap}
        \vspace{-.08cm}
       In all previous results, both the proposed and single-server approaches, categorized as JCSRA methods, consistently outperform the rigid-based formulation.
        However, the two methods have very close performance in general. This section aims to identify scenarios where a multi-server formulation is unnecessary, as the solutions closely align with those of the single-server approach, and to highlight cases where employing a multi-server setup is advantageous.
        
        In this section, we use Dirichlet distribution~\cite{hsu_dirichlet_2019} with parameter~$\alpha>0$ to characterize the system heterogeneity, specifically here regarding the amount of data samples possessed by each FL UE. 
         To alleviate other potential effects, we assume that all FL UEs are equidistant from the BS.
       The following observations can be made regarding the performance gaps between the proposed multi-server and single-server approaches:
        \begin{itemize}[leftmargin=12pt]
            \item The gaps are larger when the system is more homogeneous and energy sufficient (Fig.~\ref{fig: alpha influence}). This is because the advantage of the multi-server approach over the single-server method arises from \ac{RB}{ allocation} during uplink communication. 
            {
           	We observe in addition that all methods have their finishing time decreased when the system becomes more homogeneous, i.e., when~$\alpha$ increases, which is expected since it is known that heterogeneity increases the total latency. Furthermore, we observe that even when the system is homogeneous, rigid-based formulation is performing worse than JCSRA schemes. This is due to the coexistence design with HB traffic, which can also only be rigidly allocated, whereas DL inherently requires fewer RBs for downlink transmission and local training than for uplink; therefore, it exhibits quite poor performance.}
            \item The quality of the communication channel generally affects the size of the gap (Fig.~\ref{fig: cellR influence}). {The more challenging the channel quality is, the more the proposed schemes have advantages. In addition, when the system is heterogeneous $\alpha=0.1$, the gap between rigid allocation and JCSRA schemes increases with cell radius.} 
            \item More available energy leads to increased performance gaps  (Fig.~\ref{fig: Ecstr influence}). When energy availability is limited, it is more beneficial to allow other UEs to wait longer to save energy, resulting in only one UE occupying the RB{, hence multi-server JCSRA would approach the performance of single-server}.
        \end{itemize}

		\begin{figure*}
			\centering
	\includegraphics[width=0.25\linewidth, trim=0 -5 0 0]{final_figures_eps/j6_common_legend.eps}
			\vfill
			\begin{minipage}[t]{0.31\textwidth}
				\centering
				\includegraphics[width=\linewidth]{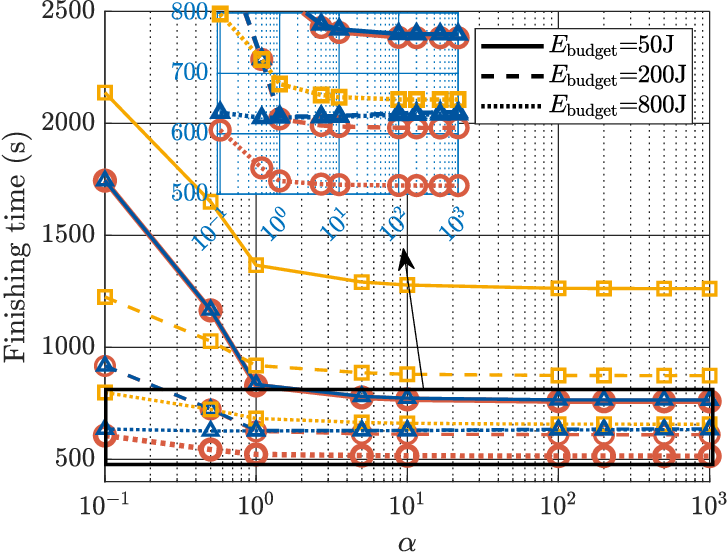}
								\vspace{-.7cm}
				\caption{Heterogeneity factor~$\alpha$ impact with FL UEs having distance of \SI{1000}{\meter} to the BS.   {Legend: Line style indicates $E_{\mathrm{budget}}$ (black lines); color/marker indicates methods.}}
				\label{fig: alpha influence}
			\end{minipage}
			\hfill
			\begin{minipage}[t]{0.31\textwidth}
				\centering
				\includegraphics[width=\linewidth]{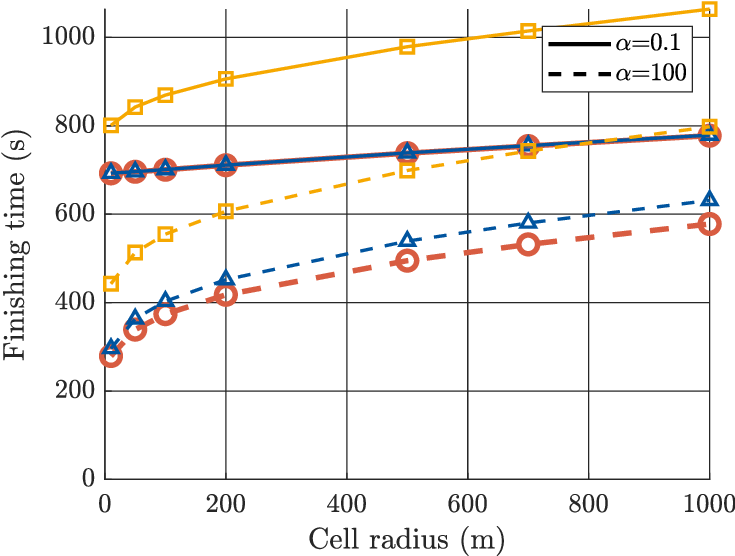}
								\vspace{-.7cm}
				\caption{Cell radius impact with $E_{\mathrm{budget}}= \SI{30}{\joule}$.   {Legend: Line style indicates $\alpha$ (black lines); color/marker indicates methods.}}
				\label{fig: cellR influence}
			\end{minipage}
            \begin{minipage}[t]{0.31\textwidth}
				\centering
				\includegraphics[width=\linewidth]{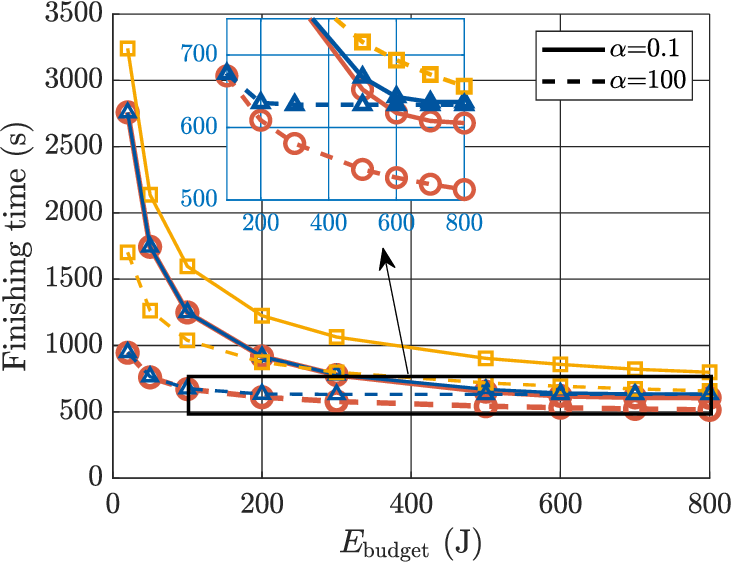}
								\vspace{-.7cm}
				\caption{Energy Constraints impact with FL UEs having distance of \SI{1000}{\meter} to the BS. {Legend: Line style indicates $\alpha$ (black lines); color/marker indicates methods.}}
				\label{fig: Ecstr influence}
			\end{minipage}
						\vspace{-.37cm}
		\end{figure*}

		\vspace{-.3cm}		
		\subsection{{Simulation under Small-Scale Fading Environment}}
			\vspace{-.1cm}
		This work assumes that \ac{CR} lasts within a large-scale coherence block time, so the previously shown results were computed with statistical-average channel states representing a time-average performance. In practice, the channels are subject to fast fading, and the RB is the smallest schedulable amount in 5G NR, so that only an integer amount can be scheduled; also, which RB should be allocated to which UE should be determined. In this subsection, we show an example of how to integrate the solutions of JCSRA into an existing simple real-time algorithm. We will show that the gap between the real-time achieved and the time-average solution is very close.
		
		We adopt the algorithm proposed in~\cite[Section IV]{Leith_FlexiblePF_2013}. The algorithm only requires as input a target rate for each UE and some hyperparameters. It consists of keeping track of the current time-average rate and updating a weight for each UE based on the gap between the time-average rate and the target rate. Then the RB is allocated to UEs depending on the UE weights and the instantaneous channel gain. 
		
		To integrate JCSRA results to the algorithm, we choose as target rate the rate of each UE at the current session given time-averaged parameters, for instance in uplink, the rate is calculated as $r_{s}^{(\mathrm{ul})}$ based on optimal $K_{s,\ell}^{(\mathrm{ul})}$ and $p_{s,\ell}^{(\mathrm{ul})}$ given by~\Cref{algo: iterative giving ordering}. In uplink, $p_{s,\ell}^{(\mathrm{ul})}$ is distributed equally to allocated RBs.  As for the sessions, in addition to \Cref{def: session}, we define that a session can move to the next session, only if the predicted amount of data to be transmitted (only for FL tasks) is attained for all involved UEs at this session.
		
		The simulation environment increments in the time step of slot size of~$\Delta$, the amount of data accumulated with time, and at the end of the simulations (all FL UEs finished their transmission tasks), all constraints are confirmed to be verified. We can see the gap between the real-time algorithm implementation and the time-average performance as shown in \Cref{fig: pareto}. We observe that the relative gap is consistently negligible below \SI{1}{\percent}, which confirms that the predicted time-average performance can be achieved with a small gap.
		
		\begin{figure}[t]
			\centering
			\includegraphics[width=0.8\linewidth]{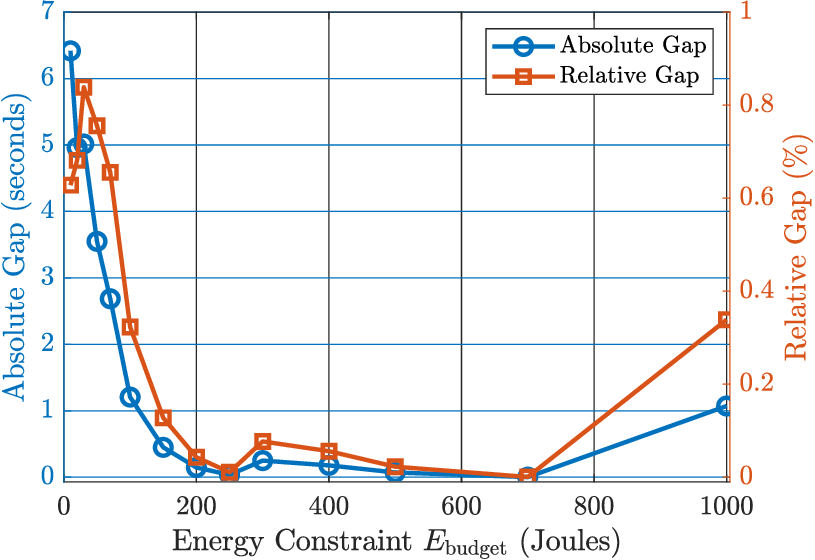}
			\caption{{Gap between real time algorithm and average performance given by JCSRA with $\Delta=\SI{100}{\milli\second}$.}}
			\label{fig: RT gap}
		\end{figure}

		\section{Conclusion}
		\vspace{-.12cm}
		In this work, we have investigated an efficient seamless integration of{ synchronous} DL services within next-generation wireless networks, particularly focusing on the coexistence of DL and HB traffic. The proposed {time-slot-wise} resource allocation framework, reformulated as a session-based problem, consists of a multi-server JCSRA problem. 
        We establish scheduling properties and feasibility conditions for the JCSRA problem and propose an iterative algorithm to tackle the non-convex and non-separable constrained optimization problem. 
		Simulation results validate the proposed method, and highlight the factors influencing the performance gap between the proposed multi-server JCSRA, single-server JCSRA, and rigid allocation. This confirms that the \emph{rigid} resource allocation within a CR is indeed inefficient and inaccurate in terms of the achievable   CR latency with tight and heterogeneous constraints under reasonable allocation. This work also shows that considering single-server JCSRA can be sufficient under constrained and heterogeneous systems, while multi-server JCSRA should be considered for better efficiency in homogeneous and resource-sufficient systems.
        Given the difficulty in quantifying the degree of heterogeneity or constraints within a given system, it remains crucial to consider multi-server JCSRA.  
        Overall, this work emphasizes the importance of considering multi-server JCSRA or other time-dependent optimization within each CR for designing efficiently and estimating accurately the consumed latency with energy constraint when enabling DL in future wireless networks. {Although the current optimization approach would be infeasible for asynchronous DL, this work demonstrates the potential efficiency gain in any practical DL system constrained by limited bandwidth and energy budget. Future work will address the design of energy-budget-aware efficient asynchronous DL frameworks based on the JCSRA principle.}{ We will also investigate in the future extending JCSRA to non-stationary channel conditions.}
		
		\vspace{-.4cm}
		\appendices
		\section{Proof~\Cref{theorem: feasibility (rigid)})}
		\label{appendix: proof feasibility rigid}
		\vspace{-.15cm}
			\begin{itemize}
				\item Increasing~$\tau_{s,cp}$ arbitrarily large makes feasible the constraint~\eqref{cons : dl/ul separation (rigid)} and the first term of~\eqref{cons: energy (rigid)}  arbitrarily small.
				\item  For any $K>0$ and any~$s$ , the function~$p\in\R_{+}\backslash\{0\}\mapsto \frac{pD}{r_{s,\mathrm{ul}}(K,p)}$ is monotonically increasing (derivative greater than zero by noticing~$\log(1+x)\geq x/(1+x)$ for any~$x>0$) and has limit of $D/(B\gamma_s)$ when~$p\to 0$. Therefore it satisfies that for any~$p\geq 0$ and $K\geq 0$,
				\begin{equation}
					(\forall p\geq 0)(\forall K\geq 0)\quad	\frac{pD}{r_{s,\mathrm{ul}}(K,p)} > \frac{D}{B\gamma_s}.
					\label{eq: bound of the transmit energy function wrt p}
				\end{equation}
			\end{itemize}
			If the feasibility condition is satisfied, then~$K'\defeq K - a\theta >0${ (it can be derived using similar principle as~\Cref{theorem: HB reformulation})} and a feasible point can be found. First take any strictly positive $\{K_{s,\mathrm{ul}}\}_s$ that satisfies the RB{ allocation} constraint, e.g., $K_{s,\mathrm{ul}}=K' / S$ for any~$s$. Then  denoting~$\varepsilon \defeq E_{s,\mathrm{budget}} - \frac{D}{B
				\gamma_s} > 0$. One can easily find~$\tau_{s,cp}$ s.t. the first term of the sum is smaller than~$\varepsilon/2$ and by monotony and the limit of the function w.r.t. $p$, where the second term is smaller than~$\varepsilon/2$. The conditions provided are therefore sufficient conditions for feasibility.
			
			If one of the conditions is not satisfied, the corresponding constraint cannot be feasible. Therefore they are also necessary. Statement proven.
		
			\vspace{-.6cm}
		\section{Proof~\Cref{theorem: feasibility (session_pb)}}
		\label{appendix: proof feasibility session}
		   		\vspace{-.15cm}
		   		
		   	\paragraph{Sufficient conditions}
			According to the following two lemmas, any feasible point of~$(\setP_{\mathrm{rig}})$ can be transformed into a feasible point of~$(\setP_{\sigma})$. Therefore, if the condition holds, $(\setP_{\mathrm{rig}})$ is feasible, so is $(\setP_{\sigma})$. The conditions are sufficient for feasibility.
			
			\begin{lemma}
				For any given feasible point~$(K_{\mathrm{\mathrm{dl}}}, K_{s,\mathrm{ul}}, p_{s,\mathrm{ul}}, \tau_{s,cp})_s \in \setX_{\mathrm{rig}}$ of~$(\setP_{\mathrm{rig}})$, a feasible point~$(\tilde{K}_{\mathrm{\mathrm{dl}}}, \tilde{K}_{s,\mathrm{ul}}, \tilde{p}_{s,\mathrm{ul}}, \tilde{\tau}_{s,cp})_s$ of~$(\setP_{\mathrm{rig}})$ with additional constraint of uplink access ordering~$\sigma$ can be found.
			\end{lemma}
			\begin{proof}
			
			Let~$(K_{\mathrm{dl}}, K_{s,\mathrm{ul}}, p_{s,\mathrm{ul}}, \tau_{s,cp})_s \in \setX_{\mathrm{rig}}$ be a feasible point of~$(\setP_{\mathrm{rig}})$. We will construct a point that satisfies~$\sigma$ ordering. To convert to a feasible point~$(\tilde{K}_{\mathrm{dl}}, \tilde{K}_{s,\mathrm{ul}}, \tilde{p}_{s,\mathrm{ul}}, \tilde{\tau}_{s,cp})_s$ given any ordering~$\sigma$, the following additional constraints have to be met without deteriorating other constraints:
\begin{equation}
		(\forall s'>s),\ \frac{D}{\tilde{r}_{\sigma(s), \mathrm{dl}}} + \tilde{\tau}_{\sigma(s), cp}
		\leq 
		\frac{D}{\tilde{r}_{\sigma(s'), \mathrm{dl}}} + \tilde{\tau}_{\sigma(s'), cp}.
		\label{eq: ordering constraint (rigid)}
\end{equation}
			This can be translated into for any~$s\in\setS$, 
	\begin{equation}
		\frac{D}{\tilde{r}_{\sigma(s), \mathrm{dl}}} + \tilde{\tau}_{\sigma(s), cp} \geq \max_{1\leq s'<s}\{	\frac{D}{\tilde{r}_{\sigma(s'), \mathrm{dl}}} + \tilde{\tau}_{\sigma(s'), cp}\}.
	\end{equation}
	We set~$\tilde{K}_{\mathrm{dl}}=K_{\mathrm{dl}}$, $\tilde{K}_{s,\mathrm{ul}}=K_{s,\mathrm{ul}}$, and $\tilde{p}_{s,\mathrm{ul}} = p_{s,\mathrm{ul}}$. As for $\tilde{\tau}_{s,cp}$, take~$\tilde{\tau}_{\sigma(1),cp}=\tau_{\sigma(1),cp}$ and then iteratively for~$s=2,\ldots, S$, such that
	\begin{multline}
\frac{D}{\tilde{r}_{\sigma(s), \mathrm{dl}}} + \tilde{\tau}_{\sigma(s), cp} = \max\Big\{\frac{D}{r_{\sigma(s), \mathrm{dl}}} + \tau_{\sigma(s), cp}, 	\\
\max_{1\leq s'<s}\{	\frac{D}{\tilde{r}_{\sigma(s'), \mathrm{dl}}} + \tilde{\tau}_{\sigma(s'), cp}\}
\Big\},
	\end{multline}
	which means
	\begin{equation}
	\tilde{\tau}_{\sigma(s), cp}\!\! =\! \max\Big\{\! \tau_{\sigma(s), cp}, \!	\max_{1\leq s'<s}\!\{\!	\frac{D}{\tilde{r}_{\sigma(s'), \mathrm{dl}}} + \tilde{\tau}_{\sigma(s'), cp}\} - \frac{D}{r_{\sigma(s), \mathrm{dl}}}
	\! \Big\}.
	\end{equation}
	
	By construction, the ordering constraint is satisfied, also $\tilde{\tau}_{s, cp}\geq \tau_{s,cp}$ for any~$s\in\setS$. Now we verify for the originally existing constraints of~$(\setP_{\mathrm{rig}})$. The constraint~\eqref{cons: uplink RB sharing (rigid)} is clearly satisfied since~$\tilde{K}_{s,\mathrm{ul}} = K_{s,\mathrm{ul}}$. The constraint~\eqref{cons : dl/ul separation (rigid)} is satisfied, since for any~$s\in\setS$, it 
	\begin{equation}
	\frac{D}{\tilde{r}_{s,\mathrm{dl}}} + \tilde{\tau}_{s,cp}\geq 	\frac{D}{r_{s,\mathrm{dl}}} + \tau_{s,cp} \geq \max_{s'}\{\frac{D}{r_{s',\mathrm{dl}}}\} = \max_{s'}\{\frac{D}{\tilde{r}_{s',\mathrm{dl}}}\}.
	\end{equation}
	For the energy constraint~\eqref{cons: energy (rigid)}, we have
	\begin{equation}
	\Big[\kappa \frac{\zeta^3\Theta_{s}^3}{\tilde{\tau}_{s,cp}^2} + \frac{\tilde{p}_{s,\mathrm{ul}}D}{\tilde{r}_{s,\mathrm{ul}}} \Big]\leq	\Big[\kappa \frac{\zeta^3\Theta_{s}^3}{\tau_{s,cp}^2} + \frac{p_{s,\mathrm{ul}}D}{r_{s,\mathrm{ul}}} \Big] \leq E_{s,\mathrm{budget}}.
	\end{equation}
	All constraints are satisfied. 
\end{proof}

\begin{lemma}
	For given ordering permutation~$\sigma$, given a feasible point~$(\tilde{K}_{\mathrm{dl}}, \tilde{K}_{s,\mathrm{ul}}, \tilde{p}_{s,\mathrm{ul}}, \tilde{\tau}_{s,cp})_s\in\setX_{\mathrm{rig}}$ of~$(\setP_{\mathrm{rig}})$ that satisfies the ordering constraint~\eqref{eq: ordering constraint (rigid)}, a feasible point to the session-based problem~$(\setP_{\sigma})$ can be found.
\end{lemma}
\begin{proof}
	Let~$\tilde{\tau}_{s, \mathrm{dl}} = \frac{D}{\tilde{r}_{s,\mathrm{dl}}}$ be the downlink finishing time of UE~$s\in\setS$.  Since the constraint~\eqref{eq: ordering constraint (rigid)} is satisfied, it is satisfied that:
	\begin{equation}
	\tilde{\tau}_{\sigma(1), \mathrm{dl}} + \tilde{\tau}_{\sigma(1), cp} \!\leq\!	\tilde{\tau}_{\sigma(2), \mathrm{dl}} + \tilde{\tau}_{\sigma(2), cp}\!\leq\! \cdots\!\leq\! \tilde{\tau}_{\sigma(S), \mathrm{dl}} + \tilde{\tau}_{\sigma(S), cp}.
	\end{equation}
	The end finishing time is calculated as
	\begin{equation}
	T \defeq	\max_{s\in\setS}\ \Big\{ \tilde{\tau}_{s, \mathrm{dl}} + \tilde{\tau}_{s, cp} + \frac{D}{\tilde{r}_{s,\mathrm{ul}}}.
		\Big\}
	\end{equation}
	Define as follows the downlink session time and idle time
	\begin{equation}
		\begin{cases}
			t_{1}^{(\mathrm{dl})}\defeq \tilde{\tau}_{1,\mathrm{dl}},\\
			(\forall s=2,\ldots,S),\ t_s^{(\mathrm{dl})} \defeq \tilde{\tau}_{s,\mathrm{dl}} - \tilde{\tau}_{s-1,\mathrm{dl}},\\
			T_{\mathrm{idle}} \defeq \min\limits_{s'}\Big\{\frac{D}{
				\tilde{r}_{s',\mathrm{dl}}}+ \tilde{\tau}_{s',cp}\Big\} - \max\limits_{s'}\Big\{\frac{D}{\tilde{r}_{s',\mathrm{dl}}}\Big\}.
		\end{cases}
	\end{equation}
	and for uplink session times
	\begin{equation}
		\begin{cases}
			(\forall s=1,\ldots,S-1),\ t_s^{(\mathrm{ul})} \defeq \tilde{\tau}_{\sigma(s+1), \mathrm{dl}} + \tilde{\tau}_{\sigma(s+1), cp} \\
			\quad\quad\quad\quad\quad - (\tilde{\tau}_{\sigma(s), \mathrm{dl}} + \tilde{\tau}_{\sigma(s), cp}), \\
			t_S^{(\mathrm{ul})} \defeq T -( \tilde{\tau}_{\sigma(S), \mathrm{dl}} + \tilde{\tau}_{\sigma(S)}).
		\end{cases}
	\end{equation}
	We define other variables equal to the corresponding rigid variables:
	\begin{equation}
		\begin{cases}
		(\forall \ell'\in\setS)\ 	K_{\ell'}^{(\mathrm{dl})} \defeq \tilde{K}_{\mathrm{dl}},\\
		(\forall s\leq\ell\in\setS)\ 	K_{s,\ell}^{(\mathrm{ul})} \defeq \tilde{K}_{\sigma(s),\mathrm{ul}},\ p_{s,\ell}^{(\mathrm{ul})} \defeq \tilde{p}_{\sigma(s),\mathrm{ul}}, \\
		(\forall \ell,\ell'\in\setS)\  K_{\mathrm{HB},\ell'}^{(\mathrm{dl})}\defeq K_{\mathrm{HB},\ell}^{(\mathrm{ul})} \defeq K'.\\
		
		\end{cases}
	\end{equation}
	One can verify that all constraints of~$(\setP_{\sigma}$) are met with this construction.
\end{proof}
The statement holds by combining both lemmas.

		\paragraph{Necessary conditions}			
			To prove the conditions necessary, prove each of them necessary. Assuming first condition does not hold, i.e., $K\leq a\theta$, due to the completion constraints, it has to hold that there exists some $K_{\ell'}^{(\mathrm{dl})}>0$ and $K^{(\mathrm{ul})}_{\sigma(s),\ell} >0$. Therefore, $K_{\mathrm{HB},\ell'}^{(\mathrm{dl})}, K_{\mathrm{HB},\ell}^{(\mathrm{ul})} < K$ for some~$\ell, \ell'\in\setS$. On the other hand, the \ac{LHS} of the constraint~\eqref{cons: eMBB each e (session_pb)} satisfies: $\mathrm{LHS of \eqref{cons: eMBB each e (session_pb)}}<KT \leq a\theta T$, which leads to contradiction with the constraint.

			Now assume the second condition does not hold, i.e., $D\geq E_{s_0,\mathrm{budget}}B\gamma_{s_0}$ for some~$s_0$. Due to the coupling of $t_{\ell}^{(\mathrm{ul})}$ for all~$s\leq\ell$ in~\eqref{cons: ul completion (session_pb)} and \eqref{cons: energy_cstr (session_pb)}, the proof method used in~\Cref{theorem: feasibility (rigid)} cannot be applied. 
            
            Let~$\varepsilon>0$. The goal is to prove that the following two conditions cannot be satisfied simultaneously for $(K_{\sigma(s),\ell}^{(\mathrm{ul})}, p_{\sigma(s),\ell}^{(\mathrm{ul})}, t_{\ell}^{(\mathrm{ul})})$:
			\begin{equation}
				\begin{cases}
					(\forall s\in\setS)\quad	\sum\limits_{\ell\geq s}r^{(\mathrm{ul})}_{\sigma(s)}(K_{\sigma(s),\ell}^{(\mathrm{ul})}, p_{\sigma(s),\ell}^{(\mathrm{ul})})t_{\ell}^{(\mathrm{ul})} \geq D, \\
					(\forall s\in\setS)\quad 	\sum\limits_{\ell\geq s }p_{\sigma(s),\ell}^{(\mathrm{ul})}t_{\ell}^{(\mathrm{ul})}  \leq E_{\sigma(s),\mathrm{budget}} - \frac{\varepsilon}{2},
				\end{cases}
			\end{equation}
			where $T_{\mathrm{idle}}$ is assumed large enough so that the computational energy term of~\eqref{cons: energy_cstr (session_pb)} is smaller than~$\varepsilon/2$ for all~$s\in\setS$.
			
			A feasible condition on~$(K_{\sigma(s),\ell}^{(\mathrm{ul})}, p_{\sigma(s),\ell}^{(\mathrm{ul})}, t_{\ell}^{(\mathrm{ul})})$ is equivalent to feasible conditions on~$t_{\ell}^{(\mathrm{ul})}$ for any given feasible~$(K_{\sigma(s),\ell}^{(\mathrm{ul})}, p_{\sigma(s),\ell}^{(\mathrm{ul})})$.
			Consider the \ac{LP} constraints w.r.t. $t_{\ell}^{(\mathrm{ul})}$ for given $r_{s,t}$ and $p_{s,t}$. We are going to prove that the LP is not feasible due to the assumed condition.
			By Farka's lemma, the problem is not feasible if we can find $a_s\geq 0$ and $b_s\geq 0$ for all $s$ s.t.
			\begin{equation}
				\begin{cases}
					(\forall \ell\in\setS)\quad	\sum\limits_{s\leq \ell}b_s p_{s,\ell} \geq \sum\limits_{s\leq\ell}  a_s r_{s,\ell}, \\
					\sum\limits_{s\in\setS}b_s (E_{s,\mathrm{budget}}-\frac{\varepsilon}{2}) < D\sum\limits_{s\in\setS} a_s.
				\end{cases}
                \label{eq: Farka's lemma}
			\end{equation}     
			Take $a_{s_0} = 1 / (B\gamma_{s_0})$ and $a_s=0$ for $s\neq s_0$, and $b_{s_0}=1$ and $b_{s}=0$ for $s\neq s_0$. 
			The condition to be satisfied for infeasibility can be written as
			\begin{equation}
				\begin{cases}
					(\forall\ell \geq s_0)\quad   p_{s_0,\ell} \geq  \frac{r_{s_0,\ell}}{B\gamma_{s_0}} ,\\
					E_{s_0,\mathrm{budget}}-\frac{\varepsilon}{2} < \frac{D}{B\gamma_{s_0}}.
				\end{cases}
			\end{equation} 
			The first inequality always holds due to~\eqref{eq: bound of the transmit energy function wrt p}. While the second inequality holds by the assumed infeasibility condition.
			
			In summary,  the infeasibility condition of Farka's lemma~\eqref{eq: Farka's lemma} is satisfied. The considered LP is not feasible for any given feasible $(K_{\sigma(s),\ell}^{(\mathrm{ul})}, p_{\sigma(s),\ell}^{(\mathrm{ul})})$ for any~$\varepsilon>0$, therefore the problem~$(\setP_{\sigma})$ is infeasible.
			The condition is necessary for feasibility. Statement proven.

		\vspace{-.4cm}
		\section{Proof~\Cref{theorem: non-preemptive_non-idle}}
		\label{appendix: proof non preemptive non idle}
		\vspace{-.2cm}

			\noindent \textbf{Non-Preemptive}: A scheduling is non-preemptive if once a UE is scheduled for communication, it does not stop until completion, i.e., for $s\leq\ell$,
			\begin{multline}
			t_{\ell}>0\wedge \Big(K_{s,\ell}=0\vee p_{s,\ell}=0\Big) \\ 
			\implies \forall \overline{\ell}> \ell, K_{s,\overline{\ell}}=0 \vee p_{s,\overline{\ell}}=0\vee t_{\overline{\ell}}=0.
			\end{multline}
			
			Consider the convex sub-problem with given  feasible non-uplink-related variables $t_{\ell'}^{(\mathrm{dl})}, t_{\ell}^{(\mathrm{ul})}, T_{\mathrm{idle}}, K_{\ell'}^{(\mathrm{dl})}, K_{\mathrm{HB},\ell'}^{(\mathrm{dl})}$:			
			\begin{subequations}
				\label{pb: subproblem given times (TfixPb)}
				\begin{align}
					\!\max_{\{K_{\sigma(s),\ell}^{(\mathrm{ul})}, p_{\sigma(s),\ell}^{(\mathrm{ul})}, K_{\mathrm{HB},\ell}^{(\mathrm{ul})}\}}\hspace{-.13cm}  & \min_{s\in\setS} \Big\{ \sum_{\ell\geq s}r^{(\mathrm{ul})}_{\sigma(s)}(K_{\sigma(s),\ell}^{(\mathrm{ul})}, p_{\sigma(s),\ell}^{(\mathrm{ul})})t_{\ell}^{(\mathrm{ul})}\Big\}
					\label{obj: (TfixPb)}
					\\
					\mathrm{s.t.}\quad\quad \  & \sum_{\ell} K_{\mathrm{HB},\ell}^{(\mathrm{ul})}t_{\ell}^{(\mathrm{ul})} \geq \theta'
					\label{cons: embb (TfixPb)}
					\\ 
					&\forall s,\ \sum_{\ell\geq s }p_{\sigma(s),\ell}^{(\mathrm{ul})}t_{\ell}^{(\mathrm{ul})}  \leq E_{\sigma(s),\mathrm{budget}}'
					\label{cons: energy cstr (TfixPb)}
					\\
					& \forall \ell,\ \sum_{s\leq \ell}K_{\sigma(s),\ell}^{(\mathrm{ul})} + K_{\mathrm{HB},\ell}^{(\mathrm{ul})}\leq K \label{cons: ul RB (TfixPb)}
					\\
					& \forall s,\ \forall\ell,\   p_{\sigma(s), \ell}^{(\mathrm{ul})}\in[0,P_{\max}] 
					\label{cons: power (TfixPb)}
					\\
					&K_{s,\ell}^{(\mathrm{ul})}, K_{\mathrm{HB},\ell}^{(\mathrm{ul})}\geq 0.
					\label{cons: positivity (TfixPb)}
				\end{align}
			\end{subequations}
			For ease of notation, here replace $\sigma(s)$ by~$s$, the solution can be easily interchanged afterwards.
			
			Introduce~$u>0$ to establish its epigraph equivalent form. The Lagrangian of the resulting problem can be written as:
					\begin{multline}
					\bigL =\! -u \!+\! \sum_s\!\lambda_s(u \!-\! \sum_{\ell\geq s} r_{s,\ell}t_{\ell}) + \gamma (\theta' - \sum_{\ell}K_{\mathrm{HB},\ell}t_{\ell}
				)
					\\[-.1cm]
				+ \sum_s a_s (\sum_{\ell\geq s}p_{s,\ell} t_{\ell} - E_{s,\mathrm{budget}}') + \sum_{\ell} \mu_{\ell} (\sum_{s\leq\ell} K_{s,\ell}+ K_{\mathrm{HB},\ell}-K) 		\\[-.1cm]
				+\! \sum_{s\leq \ell}\!\Big(\!-\alpha_{s,\ell}K_{s,\ell}  -\underline{\theta}_{s,\ell}p_{s,\ell}+\overline{\theta}_{s,\ell}(p_{s,\ell} - P_{\max}) \!\Big)-\sum_{\ell}\beta_{\ell}K_{\mathrm{HB},\ell}.\\[-.5cm]
		\end{multline}
			The first-order stationary condition can be written as:
			\begin{numcases}{}
				\sum\limits_s \lambda_s = 1 \label{eq: sum lambda (proof)}\\[-.2cm]
				(\forall s,\ell \ s.t.\ s\leq \ell) \ 	\mu_{\ell} = \alpha_{s,\ell}+\lambda_st_{\ell}\frac{\partial r_{s,\ell}}{\partial K} \label{eq: der K (proof)}\\
				(\forall \ell)\ 	\mu_{\ell} = \beta_{\ell} + \gamma t_{\ell}  \label{eq: der K_HB (proof)} \\
				(\forall s,\ell \ s.t.\ s\leq \ell)\	\lambda_st_{\ell} \frac{\partial r_{s,\ell}}{\partial p} + \underline{\theta}_{s,\ell} = a_st_{\ell} + \overline{\theta}_{s,\ell}.  \label{eq: der p (proof)}
			\end{numcases}
			Since the optimization problem is convex and the Slater's condition holds (if the given variables are feasible, then there exists a feasible point of the subproblem where $u^*\geq D$; take $u= D/2$, then there exists necessarily a strictly feasible point), the \ac{KKT} conditions are sufficient and necessary conditions for optimality. From \ac{KKT}, we derive certain conditions on the optimal solutions.
			\vspace{-.15cm}
			\begin{lemma}
				If $K_{s_0,\ell_0} = 0$ for some $s_0,\ell_0 \in\setS$ and $t_{\ell_0}>0$, then necessarily, $\lambda_{s_0} = 0$.
                \vspace{-.15cm}
			\end{lemma}
            \vspace{-.15cm}
			\begin{proof}
				$\frac{\partial r_{s,\ell}}{\partial K}$ tends to infinity when $K\to 0$, for~\eqref{eq: der K (proof)} to satisfy, the statement has to hold.
                \vspace{-.15cm}
			\end{proof}
            \vspace{-.15cm}
			\begin{lemma}
				If $\lambda_{s_0}=0$ for some $s_0\in\setS$, there has to exist $\ell_1\geq s_0$ such that $\mu_{\ell_1}=0$ and $t_{\ell_1}>0$.
				 In this case, $\forall s\leq \ell_1$, $\lambda_s=\alpha_{s,\ell_1}=0$ for any~$t_{\ell}>0$. Specifically, $\forall s\leq s_0$, $\lambda_s=0$.
				\label{lemma: lower lambda zero}
                \vspace{-.15cm}
			\end{lemma}
            \vspace{-.5cm}
			\begin{proof}
				There has to exist, since otherwise for any~$\ell\geq s_0$, $\mu_{\ell} = \alpha_{s_0,\ell}>0$ or $t_{\ell}=0$, therefore $K_{s_0,\ell}=0$ or $t_{\ell}=0$ for all $\ell\geq s_0$ by complementary slackness, which cannot happen due to the completion constraint in the primal problem.
				
				Take such~$\ell_1$. By~\eqref{eq: der K (proof)} and $\frac{\partial r_{s,\ell}}{\partial K}>0$, we have the statement.
					\vspace{-.25cm}
			\end{proof}
			\vspace{-.15cm}
			\begin{lemma}
				There exists~$\overline{s}\in\setS$ s.t. $\forall s< \overline{s}$, $\lambda_s =0$, and $\forall s\geq \overline{s}$, $\lambda_s>0$. Therefore,
				\begin{equation}
					(\forall s<\overline{s}) (\forall\ell\geq\overline{s})\ K_{s,\ell}=0,
				\end{equation}
				and $(\forall s\geq\overline{s}) (\forall\ell\geq s)\ K_{s,\ell}, p_{s,\ell}>0$ if $t_{\ell}>0$.
				\label{lemma: property on strict positivity K}
				\vspace{-.1cm}
			\end{lemma}
			\begin{proof}
				From~\eqref{eq: sum lambda (proof)}, we know that not all~$\lambda_s=0$. By~\Cref{lemma: lower lambda zero}, the existence of such~$\overline{s}$ has to hold by taking the largest~$s$  such that $\lambda_s=0$.
				 As for the consequence on~$K_{s,\ell}$, applying~\eqref{eq: der K (proof)} on $\ell\geq\overline{s}$ and any $s\geq\overline{s}$ s.t. $t_{\ell}>0$ yields $\mu_{\ell}>0$. The two statements follow by complementary slackness and~\eqref{eq: der K (proof)}. 
	\vspace{-.15cm}
			\end{proof}
			From~\Cref{lemma: property on strict positivity K}, for~$s\geq\overline{s}$, for any~$\ell\geq \overline{s}$, $K_{s,\ell}>0$, therefore non-preemptive. For all~$s<\overline{s}$, and $\ell\in\{s,\ldots, \overline{s}\}$,  since~$\lambda_s=0$ and $\mu_{\ell}=\alpha_{s,\ell}$, so $\mu_{\ell}=0$ if $K_{s,\ell}>0$. Therefore, no constraint has to be active as long as they are feasible to the primal constraints. 
			We construct a non-preemptive and feasible $K_{s,\ell}$ and $p_{s,\ell}$ for $s<\overline{s}$. 
			We have proven the non-preemptiveness of an optimal solution.
			
			\vspace{-.15cm}
			\begin{lemma}
				If $\overline{s}$ from~\Cref{lemma: property on strict positivity K} satisfies $\overline{s}\geq 2$, then $\gamma=0$ and $K_{\mathrm{HB},\ell}=0$ for~$\ell\geq \overline{s}$.
				\label{lemma: s geq 2}
				\vspace{-.1cm}
			\end{lemma}
			\begin{proof}
				Since~$\overline{s}\geq 2$, there exists $s$ s.t. $\lambda_s=0$. It has to hold that there exists~$\mu_{\ell}=0$. Therefore, $\gamma=0$.
				Then for $\ell$ s.t. $\mu_{\ell}>0$, $\beta_{\ell}>0$ holds. The statement holds by complementary slackness.
					\vspace{-.15cm}
			\end{proof}
			The~\Cref{lemma: s geq 2} indicates that $\overline{s}\geq 2$ happens only when all constraints w.r.t. $s,\ell< \overline{s}$ and the HB traffic constraint can already be satisfied with $K_{\mathrm{HB},\ell}=0$ for~$\ell\geq \overline{s}$.

			\noindent\textbf{Non-idle}:
			We prove for the uplink, the downlink can be proven similarly.
			Suppose there exists an optimal solution that has idle communication session time in the uplink phase, i.e., for some~$s_0,\ell_0\in\setS$, $K_{s_0,\ell_0}^{(\mathrm{ul})} =p_{s_0,\ell_0}^{(\mathrm{ul})} =0$ with $t_{\ell_0}^{(\mathrm{ul})}>0$. Consider a point with~$T_{\mathrm{idle}}' = T_{\mathrm{idle}}+t_{\ell_0}^{(\mathrm{ul})}$ and $t_{\ell_0}^{(ul)'} = 0$. Less energy is being consumed. A strictly better point can be easily found.

			\vspace{-.35cm}
		\section{Proof~\Cref{theorem: HB reformulation}}
		\label{appendix: proof HB reformulation}
		\vspace{-.1cm}
					During uplink or downlink, given $K_e$ number of RB for each user $e$ and a slack variable~$K_{\mathrm{HB}}$ the total amount of RBs that can be used by HB traffic, a subproblem w.r.t. $K_e$ can be isolated :
			\begin{subequations}
				\begin{align}
					\max_{\{K_e\}_{\forall e\in\setE}}\quad & \min_{e\in\setE} r_e = K_eB\log(1+\gamma_e)\\
					\mathrm{s.t.}\quad& \sum_{e\in\setE} K_e \leq K_{\mathrm{HB}};\ (\forall e\in\setE)\ K_e\geq 0.
                    \label{cons: resource sharing (HB pb)}
				\end{align}
			\end{subequations}
			The epigraph form is written as:
			\begin{subequations}
				\begin{align}
					\max_{K_e, \forall e\in\setE, u}\quad & u \\
					\mathrm{s.t.}\quad& 
					\eqref{cons: resource sharing (HB pb)},\ (\forall e\in\setE)\ K_eB\log(1+\gamma_e) \geq u.
				\end{align}
			\end{subequations}
			By writing down the Lagrangian and finding the Lagrange dual function, the dual problem has the following form:
			\begin{subequations}
				\begin{align}
					\min_{\mu}\quad & + \mu K_{\mathrm{HB}} \\
					\mathrm{s.t.}\quad& \mu \geq 0,\ \mu \sum_e\frac{1}{B\log(1+\gamma_e)}=1.
				\end{align}
			\end{subequations}
			The dual optimum is attained at $\mu^* = {\frac{1}{a}}\defeq 1/ \sum_e\frac{1}{B\log(1+\gamma_e)}$ equal to $K_{\mathrm{HB}}/ \sum_e\frac{1}{B\log(1+\gamma_e)}$. The epigraph equivalent problem is \ac{LP}, therefore by strong duality, the result can be derived.
			
		\vspace{-.4cm}
		\section{Proof~\Cref{theorem: convex TP}}
		\label{appendix: proof convex problem}
		\vspace{-.2cm}
			The objective is linear  (sum of variables).
			The constraints are either affine or can be proven convex by noticing the following classic convex properties:
			\begin{itemize}[leftmargin=11pt]
				\item $(p,K)\mapsto r^{(\mathrm{ul})}_{s,\ell}(p,K)$ is a joint concave function since it is the perspective function of~$x\mapsto\log(1+ax)$ for a certain constant~$a>0$.
				\item Square root of concave functions remains concave.
				\item Convex function composed with affine function is convex.
			\end{itemize}
	\section{{Proof~\Cref{theorem: algo convergence}}}
	\label{appendix: proof algo convergence}
	Such algorithm is referred to as successive convex approximation, sequantial convex programming or inner approximation algorithm. 
	
	It has been proven in~\cite[Lemma~2.2]{SCA_Beck_2010} that the algorithm produces non-increasing objective value sequence. Since the objective value is clearly lower bounded (by zero). The sequence converge (also supported by~\cite[Corollary~2.3]{SCA_Beck_2010}).
	
	The fact that it converges to a stationary point (or so-called KKT point) of the optimization problem has been shown in~\cite[Theorem~1]{SCA_MarksWright_1978}. 
	
	Now we have to show that all assumptions of the Theorem~1 holds.
	
	\begin{itemize}[leftmargin=11pt]
		\item 	The objective and (16), (18), (20), (24b), (24c) are convex. 
		\item The objective and all constraints are differentiable.
		\item The feasible region is clearly compact. All constraints are closed, and all variables except the time variables $t_{\ell'}^{(\mathrm{dl)}}$, $t_{\ell}^{(\mathrm{ul)}}$, and $T_{\mathrm{idle}}$ are bounded. These variables have their sum to be minimized. Take $T^{(0)}$ as the initial feasible point objective value. We can bound them as follows $t_{\ell'}^{(\mathrm{dl)}}, t_{\ell}^{(\mathrm{ul)}}, T_{(\mathrm{idle)}}\in[0, T^{(0)}]$ for any $\ell',\ell$.
		\item The approximated constraints and optimization problem are convex, proven in~\Cref{theorem: convex TP}.
		\item Three conditions on the approximation: majorization, tightness, and equal gradient at the equality point:
		\begin{enumerate}
			\item For the approximation~\eqref{eq: MM surrogate def}, the majorization and equality is clear. Let's prove that the gradients are equal at the equality point.
			For ease of notation, denote $a=p_{\sigma(s),\ell}^{(ul)}$ and $b=t_{\ell}^{(ul)}$ and $c=y_{s,\ell}^{(E)}=\hat{a}/\hat{b}$. We have   \begin{equation}
			\frac{\partial (ab)}{\partial a}(\hat{a},\hat{b}) = \hat{b},\quad \frac{\partial (ab)}{\partial b}(\hat{a},\hat{b}) = \hat{a}.
			\end{equation}
			and on the other side:
			\begin{equation}
		\hspace{-.5cm}	\frac{\partial}{\partial a}( \frac{a^2}{2c}  + \frac{b^2c}{2} )(\hat{a},\hat{b}) \!=\!  \frac{\hat{a}}{c} \!=\! \hat{b}, \, \frac{\partial}{\partial b}( \frac{a^2}{2c}  + \frac{b^2c}{2} )(\hat{a},\hat{b}) \!=\!  \hat{b}c \!=\! \hat{a}, 
			\end{equation}
			We have therefore equal gradient at the equality point.
			\item For the quadratic transform~\eqref{eq: quadratic transform}, the wished property has not been proven in~\cite{Shen_FP_quadraticTransform_2018}. We will prove it for our specific case, but noted that the property also holds for arbitrary fractional programming. We need to prove that  at $y=\sqrt{X(\hat{x})}\hat{t}$, for any feasible $x$, 
			\begin{equation}
			X(x)t \geq 2y\sqrt{X(x)} - \frac{y^2}{t}.
			\end{equation}
			If $t=0$ the right hand side is not defined ($-\infty$). Now assume $t>0$,
			\begin{equation}\begin{aligned}
				& X(x) \geq 2y\frac{\sqrt{X(x)}}{t} - \frac{y^2}{t^2},\\
				\Leftrightarrow\quad &  -\frac{y^2}{t^2}  + 2y\frac{\sqrt{X(x)}}{t} -  X(x) \leq 0.      
			\end{aligned}
			\end{equation}
			The determinant of the second-order polynomial is:
			\begin{equation}
			\Delta = \frac{4X(x)}{t^2} - 4 \frac{X(x)}{t^2}=0.
			\end{equation}
			The majorization inequality holds therefore for all $x,t,y>0$ and has equality only at $y=\sqrt{X(\hat{x})}\hat{t}$. Now prove that the gradients are equal at the equality point.
			 Let's prove that \begin{equation}
			\nabla (X(x)t)|_{(\hat{x},\hat{t})} = \nabla \hat{g}(\hat{x},\hat{t},y).
			\end{equation}
			We have
			\begin{equation}
			\nabla (X(x)t)|_{(\hat{x},\hat{t})} = \begin{pmatrix}
				X'(\hat{x})\hat{t}\\
				X(\hat{x})
			\end{pmatrix}.
			\end{equation}
			On the other side:
			\begin{equation}
			\hspace{-.4cm}\nabla \hat{g}(\hat{x},\hat{t},y) \!=\!  \begin{pmatrix}
				\frac{y X'(\hat{x})}{\sqrt{X(\hat{x})}}\\
				\frac{y^2}{\hat{t}^2}
			\end{pmatrix} \!=\! \begin{pmatrix}
				X'(\hat{x})\hat{t}\\
				X(\hat{x})
			\end{pmatrix} \!=\! \nabla (X(x)t)|_{(\hat{x},\hat{t})}.
			\end{equation}
			We have therefore equal gradient at $(\hat{x}, \hat{t})$.
		\end{enumerate}
		The three conditions hold for both transforms. The inequality consists of sum of terms involving these transforms, all the properties can therefore be inherited to the constraints.
		\item All subproblems verify the Slater's conditions. The outline of the proof is as follows: In fact, the solutions of each step is a feasible point to the next subproblem, and is guaranteed to be strictly feasible at the approximated constraints by \Cref{lemma: approximated_non_active}; Then, we will prove that we can find a strictly feasible point from it by \Cref{lemma: linear independent implies strict feasible}.
		
		For each of notation, consider any constraint to be approximated as function $g$ (with constraint $g\leq 0$) and the approximation function as $\hat{g}:x\mapsto\hat{g}(x;x^{(k)})$ with $x^{(k)}$ the solution of the subproblem at $(k-1)$-th iteration, since the auxiliary variables $y\in\setY$ depends on the solution of the previous step. The function satisfies the following three conditions:
		\begin{itemize}
			\item \textit{Majorization surrogate}: $g(x)\leq \hat{g}(x; x^{(k)})$ for all $x\in\setF^{(k)}$, with $\setF^{(k)}$ the feasible set of the subproblem~$k$.
			\item \textit{Tight at $x^{(k)}$}: $g(x^{(k)}) = \hat{g}(x^{(k)}; x^{(k)})$ 
			\item \textit{The gradients are equal at $x^{(k)}$}: $\nabla g (x^{(k)}) = \nabla \hat{g}(x^{(k)}; x^{(k)})$.
		\end{itemize}
		In addition, we can easily verify that strict inequality can be achieved at all points except the equality points, i.e.,
		\begin{equation}
\forall x\in\setF^{(k)}\backslash \{x^{(k)}\}, g(x) < \hat{g}(x;x^{(k)}),
		\end{equation}
		
		 \begin{lemma}
			If $x^{(0)}$ is strictly feasible and that any constraint~$g,\hat{g}(\cdot,x^{(k)})$,  in addition, verifies the strictly inequality:
			\[
			\forall i=\ell+1,\ldots, m,\ \forall x\in\setF^{(k)}\backslash \{x^{(k)}\}, g(x) < \hat{g}(x;x^{(k)}),
			\]
			then $x^{(k)}$ obtained by Algorithm~1 for any $k\in\N$ is a strictly feasible point at this approximated constraint $g$. 
			\label{lemma: approximated_non_active}
		\end{lemma}
		\begin{proof}
			Let's prove by induction. The property clearly holds for $k=0$.\\
			Consider $k\in\N$, assume $x^{(k)}$ is a strictly feasible point of $(\setP^{(k)})$, let's prove that $x^{(k+1)}$ is a strictly feasible point of $(\setP^{(k+1)})$ of constraints $i=\ell+1,\ldots,m$.
			\\
			$x^{(k+1)}$ is a feasible point of $(\setP^{(k)})$ by definition. We have therefore
			\[
			\forall i=\ell+1\ldots,m,\quad \hat{g}(x^{(k+1)}; x^{(k)}) \leq 0.
			\]
			By the strict inequality, we have
			\[
			g(x^{(k+1)})< \hat{g}(x^{(k+1)}; x^{(k)}) \leq 0.
			\]
			We further have $g(x^{(k+1)}) = \hat{g}(x^{(k+1)}; x^{(k+1)}) < 0$. Induction proven, hence the statement holds.
		\end{proof}
		The condition holds for the proposed constraints convex approximation. Now we state the following lemma that guarantee the existence of a strictly feasible point.
		
		  \begin{lemma}
			For any point $x^*\in\setF$, denote the set of active constraints as $A(x^*) = \{i=1,\ldots,m \mid g_i(x^*)=0\}$. If $D\defeq (\nabla g_i(x^*))_{i\in A(x^*)}$ has linearly independent columns then there exists a strictly feasible point $\overline{x}$ of $\setF$.
			\label{lemma: linear independent implies strict feasible}
		\end{lemma}
		\begin{proof}
			The Gordan's lemma states that exactly one of the following statements has to happen: \begin{itemize}
				\item There exists $d\in\R^n$, s.t., $D^{\top}d <0$,
				\item There exists $y\neq 0$ and $y\geq 0$ (all elements greater than 0), s.t. $Dy =0$.
			\end{itemize}
			Since $D$ has linearly independent columns, there exists no nontrivial zero vectors such that $Dy=0$. Therefore, there has to exist $d\in\R^n$, s.t., $D^{\top}d <0$.
			
			We have therefore by definition of $D$, there exists $d\in\R^n$, \[
			\forall i\in A(x^*),\ \nabla g_i(x^*)^{\top}d<0.
			\]
			
			For each active constraint $i\in A(x^*)$, $g_i(x^*)=0$, and by Taylor's formula:
			\[
			g_i(x^*+\varepsilon d) = \varepsilon \underbrace{\nabla g_i(x^*)^{\top}d}_{<0} + o (\varepsilon).
			\]
			There exists therefore small enough $\varepsilon >0$, such that for any $i\in A(x^*)$, $g_i(x^*+\varepsilon d) < 0$.
			
			For all inactive constraints, by the continuity of $g_i$, there exists an open neighborhood of $x^*$, $V(x^*)$ such that $g_i(x)<0$ for any $x\in V(x^*)$. For small enough $\varepsilon^*$, $x^*+\varepsilon^* d \in V(x^*)$. Take the $\overline{\epsilon} = \min\{\varepsilon, \varepsilon^*\}$. We have found a strictly feasible point $x^*+\overline{\varepsilon}d$.
		\end{proof}
	All non-approximated constraints are affine, and each variable is subject to a nonnegativity bound and in exactly one other affine constraint. Moreover, for any of these remaining constraints to be active, at least one of its participating variables must be strictly positive.
	Under these conditions, the gradients of all active constraints are linearly independent.
		
		 To conclude, by Lemma~\ref{lemma: approximated_non_active}, the result from the previous iteration $x^{(k)}$ is strict (non active) at all approximated constraints involved in sequential convex approximation.
		All other constraints have their gradients linearly independent, by Lemma~\ref{lemma: linear independent implies strict feasible}, there exists a strict feasible point of the problem, therefore, the Slater's condition is satisfied for all $(\setT\setP_{\sigma}^{(k)})$.
	\end{itemize}
	To conclude, all requirements are satisfied, the algorithm is guaranteed to converge to a stationary point of the session-based optimization problem~$(\setP_{\sigma})$.
		\bibliographystyle{IEEEtran}
		\bibliography{related_works_final.bib}

	\end{document}

%% file: my_macros.tex
\usepackage[ruled,vlined]{algorithm2e}
\usepackage{amsmath}
\usepackage{amsfonts}
\usepackage{amssymb}
\usepackage{amsthm}
\usepackage{dsfont}
\usepackage{array}
\usepackage{multirow}
\usepackage{lipsum}
\usepackage{mathtools}
\usepackage{cuted}
\usepackage[caption=false]{subfig}
\usepackage{bbm}
\usepackage{siunitx}
\usepackage{bm}
\usepackage{booktabs}
\usepackage{diagbox}

\usepackage{xcolor}

\usepackage{stfloats}
\usepackage{rotating} 
\usepackage{cases}

\usepackage{hyperref}
\usepackage{cleveref}
\usepackage{lineno}

\newcommand{\setS}{\mathcal{S}}

\newcommand{\setT}{\mathcal{T}}

\newcommand{\setF}{\mathcal{F}}

\newcommand{\setE}{\mathcal{E}}
\newcommand{\setX}{\mathcal{X}}
\newcommand{\setY}{\mathcal{Y}}

\newcommand{\bigL}{\mathcal{L}}

\newcommand{\setP}{\mathcal{P}}

\DeclareSIUnit \belm {Bm}
\newcommand{\numdBm}[1]{\qty{#1}{\deci\belm}} %[per-mode = symbol]
\newcommand{\N}{\ensuremath{\mathbb{N}}}

\newcommand{\R}{\ensuremath{\mathbb{R}}}  % Real numbers
\usepackage{dsfont}

\newcommand{\defeq}{\ensuremath{\triangleq}} 

\Crefname{equation}{Eq.}{Eqs.}
\Crefname{figure}{Fig.}{Figs.}
\usepackage{acronym}

%% file: tables/notation_table.tex
\begin{table*}[htbp]
	\centering
	\caption{{Parameters and Variables}}
	\begin{tabular}{|p{2.3cm}|p{5.7cm}|p{2cm}|p{5.7cm}|}
		\hline
		\textbf{Parameter} & \textbf{Description} & \textbf{Parameter} & \textbf{Description} \\
		\hline
		\multicolumn{4}{|l|}{\textbf{System settings}} \\
		\hline
		$\Delta$ & TTI slot length (in second) & $\setE$ & Set of eMBB users \\
		$\setF$ & Set of FL users & $\setS$ & Selected set of FL users in the \ac{CR} of size $S$ \\
		$K$ & Total number of RB & $P^{(\mathrm{dl})}$ & BS downlink power for each subcarrier \\
		$P_{\max}$ & UE uplink maximum power & $N_0$ & AWGN noise level \\
		$B$ & Bandwidth of a RB ($\SI{}{\hertz}$) & $\alpha$ & Dirichlet parameter of data quantity heterogeneity of UEs (in Sec~\ref{sec: simulation gap})  \\
		\hline
		\multicolumn{4}{|l|}{\textbf{Communication}} \\
		\hline
		$h_e^{(t)}$ & Channel gain of HB user~$e$ at {time slot}~$t$ & $h_s^{(t)}$ & Channel gain of FL user~$s$ at {time slot}~$t$ \\
		$\gamma_e^{(t)}$ & power-normalized SNR of HB UE~$e$ at time~$t$ & $\gamma_{s}^{(t)}$ & power-normalized SNR of FL UE~$s$ at time~$t$ \\
		$p_{s,\mathrm{ul}}^{(t)}$ & Uplink transmit power of FL UE~$s$ at {time slot}~$t$ & $K_{e}^{(t)}$ & Number of RBs assigned to HB UE~$e$ \\
		$K_{\mathrm{dl}}^{(t)}$ & Number of RBs assigned to FL downlink broadcasting & $K_{s,\mathrm{ul}}^{(t)}$ & Number of RBs assigned to FL UE~$s$ for uplink \\
		$r_e^{(t)}$ & Rate of HB UEs~$e$ at {time slot}~$t$ & $r_{s,\mathrm{dl}}^{(t)}, r_{s,\mathrm{ul}}^{(t)}$ & Rate of FL downlink/uplink UEs at {time slot}~$t$ \\
		$\theta$ & Target HB minimum rate & & \\
		\hline
		\multicolumn{4}{|l|}{\textbf{Federated Learning}} \\
		\hline
		$D$ & Model parameter size (bits) & $\tau_{s,\mathrm{dl}}, \tau_{s,\mathrm{ul}}$ & Duration for FL UE~$s$ to finish downlink/uplink transmission \\
		$I_s$ & Number of local epochs of training & $C_s$ & CPU cycles required for training one sample data at UE~$s$ \\
		$\Theta_s$ & Local dataset sample size of UE~$s$ & $f_s\in(0,f_{s,\max}]$ & Computational capacity of UE~$s$ \\
		$\tau_s^{\!(\mathrm{cp})}\!\!\!\in\![\tau_{s,\!\min}^{\!(\mathrm{cp})},\!+\!\infty)$ & Duration for UE~$s$ to finish computation task ($\tau_{s,\min}^{\!(\mathrm{cp})}$ is calculated by $f_{s,\max}$ via~\eqref{eq: tau_cp}) & $\kappa$ & Effective switched capacitance \\
		$E_s^{(\mathrm{cp})}, E_s^{(\mathrm{cm})}$ & Energy consumed on local computation (resp. FL communication) & $T$ & Duration of the CR (seconds) \\
		$E_{s,\mathrm{budget}}$ & Energy budget of FL UE~$s$ for the CR & $E_{s}^{(\mathrm{tot})}$ &  Total consumed energy by FL UE~$s$ in the CR \\
		{$\zeta = I_sC_s$} & constant of computation. &&\\
		\hline
		\multicolumn{4}{|l|}{\textbf{Session-based Reformulation}} \\
		\hline
		$\sigma\in\mathfrak{S}$ & Starting time ordering of FL uplink & $\gamma_s, \gamma_e$  &  Statistical average of $\gamma_s^{(t)}$, $\gamma_e^{(t)}$ respectively \\
		$t_{\ell'}^{(\mathrm{dl})}$ & Duration of $\ell'$-th downlink session  &$p_{s,\ell}^{(\mathrm{ul})}$ & Uplink transmit power of FL UE~$s$ at session~$\ell$ \\
		$t_{\ell}^{(\mathrm{ul})}$ & Duration of $\ell$-th FL uplink session &  $K_{e,\ell'}^{(\mathrm{dl})}, K_{e,\ell}^{(\mathrm{ul})}$ & Average number of RBs assigned to HB UE~$e$ \\
		$T_{\mathrm{idle}}$ & Communication-idle time between uplink and downlink 
		& $K_{\ell'}^{(\mathrm{dl})}$ & Average number of RBs assigned to FL downlink broadcasting \\
		$K_{s,\ell}^{(\mathrm{ul})}$ & Number of RBs assigned to FL UE~$s$ for uplink session~$\ell$ &  $K_{\mathrm{HB}, \ell'}^{(\mathrm{dl})}, K_{\mathrm{HB}, \ell}^{(\mathrm{ul})}$ & Overall RB needed by all HB UEs in~$\setE$ during each session~$\ell',\ell$\\
		$r_{s}^{(\mathrm{dl})}(K_{\ell'}^{(\mathrm{dl})})$ & Rate of FL UE~$s$ given $K_{\ell'}^{(\mathrm{dl})}$ RB & $r_{e,\ell'}, r_{e,\ell}, r_{e,\mathrm{idle}}$ & UE-$e$'s rate at session~$\ell'$, $\ell$ or idle session respectively  \\
		$ r_{s}^{(\mathrm{ul})}(K_{s,\ell}^{(\mathrm{ul})}, p_{s,\ell}^{(\mathrm{ul})})$ & Uplink rate of UE~$s$ given~$K_{s,\ell}^{(\mathrm{ul})}$ RBs and $p_{s,\ell}^{(\mathrm{ul})}$ of transmit power  & $a$ & Constant for HB UEs defined in~\Cref{theorem: feasibility (rigid)} \\
		\hline
		
		\multicolumn{4}{|l|}{\textbf{Optimization Problem and Algorithm Development}} \\
		\hline
		$(\setP_{\mathrm{rig}})$ & Rigid resource allocation problem~\eqref{pb: rigid problem} & $(\setP_{\sigma})$ & Session-based problem defined in~\eqref{pb: session-based problem}  \\
		$y_{\mathrm{HB},\ell'}^{(\mathrm{dl})}, y_{\mathrm{HB},\ell}^{(\mathrm{ul})}$ & Auxiliary variable for decoupling product of variables for HB constraint at session $\ell'$, $\ell$, resp. & $y_{\ell'}^{(\mathrm{dl})}, y_{s,\ell}^{(\mathrm{ul})}$ & Auxiliary variable for decoupling product of variables for downlink/uplink completion constraints \\
		$y_{s,\ell}^{(E)}$ & Auxiliary variable in MM approximation & $\hat{\phi}_{s,\ell}$, $\hat{E}_s^{(\mathrm{tot})}$ & MM surrogate function defined in~\eqref{eq: MM surrogate def} and resulting spent energy expression of UE~$s$ \\
		$(\setT\setP_{\sigma})$ & Transformed session-based optimization problem defined in~\eqref{pb: session-based problem: convex iterative subproblem} &$(\setP 1_{\sigma})$ & HB constraint simplified session-based problem defined in~\Cref{theorem: HB reformulation} \\
		$\setX$, $\setX_{\mathrm{rig}}$ & Optimization variable set of session-based (resp. rigid) problem & $\setY$ &  Auxiliary variable set defined in~\Cref{sec: algorithm}\\
		$X_n, Y_n$ & Optimization variable given at iteration~$n$ of Alg.~\ref{algo: iterative giving ordering} & $T_n$ & Objective value attained at iteration~$n$ of Alg.~\ref{algo: iterative giving ordering}  \\
		\hline
	\end{tabular}	    
\end{table*}

%% file: tables/system_params_j_final.tex
\begin{table}[t]
%	\scriptsize
%	\vspace{-.4cm}
	\centering
	\caption{Parameter values used in simulations}
	\label{tab: param}
	\vspace{-.1cm}
	\begin{tabular}{|c|c||c|c|}
		\hline
		\textbf{Parameter} & \textbf{Value} &\textbf{Parameter} & \textbf{Value}\\
		\hline
		$S$ & 10 & $K$ & {18} \\ \hline
		$N_0$ & \SI[per-mode = symbol]{-174}{\deci\belm\per\hertz}& $|\setE|$ & 20\\\hline
		$P^{(d)}$& \numdBm{30} & $\theta$ & \SI[per-mode = symbol]{10}{\mega\bit\per\second}\\\hline
		$P_{\max}$ & \numdBm{23} &  $B$ & \SI{720}{\kilo\hertz}\\\hline
		freq & \SI{3.5}{\giga\hertz} &  $D$ & \SI{800}{\mega\bit}\\\hline
		{$\kappa$} &  {$10^{-28}$}&  {$f_{\max}$} &  {2 GHz}\\ \hline
		{$C_s$} &  15x32x32x3x32 &  {$I_s$} &  {20} \\ \hline 
	\end{tabular}
\end{table}

%% file: tables/res_heuristic_mr_final.tex
\begin{table}[t]
	\caption{{Percentage of random realizations with relative gap of heuristic ranking and optimal ranking below a given threshold w.r.t. different system parameters with~$S=7$}}
	\Ccolor
	\begin{tabular}{c|ccc}
		\toprule
		\multirow{2}{*}{Threshold} & \multicolumn{3}{c}{$E_{\textrm{budget}}$ (\SI{}{\joule})} \\
		& $50$ & $100$ & $300$ \\
		\hline
		$\leq 0.1 \%$ & $90 \%$ & $90 \%$ &$60 \%$\\ 
		$\leq 0.5 \%$ & $\bm{100} \%$ & $\bm{100} \%$ &$\bm{100} \%$\\ 
		$\leq 1.0 \%$ & $\bm{100} \%$ & $\bm{100} \%$ &$\bm{100} \%$\\ 
		%				\bottomrule
	\end{tabular}
	\begin{tabular}{c|ccccc}
		\toprule
		\multirow{2}{*}{Threshold} & \multicolumn{5}{c}{$\theta$ (\SI[per-mode = symbol]{}{\mega\bit\per\second})} \\
		&$0.1$ & $0.5$ & $1$ & $5$ & $10$  \\
		\hline
		$\leq 0.1 \%$ & $\bm{100} \%$ & $\bm{100} \%$ &$\bm{100} \%$ &$90 \%$ &$50 \%$\\ 
		$\leq 0.5 \%$ & $\bm{100} \%$ & $\bm{100} \%$ &$\bm{100} \%$ &$90 \%$ &$\bm{100} \%$\\ 
		$\leq 1.0 \%$ & $\bm{100} \%$ & $\bm{100} \%$ &$\bm{100} \%$ &$\bm{100} \%$ &$\bm{100} \%$\\ 
	\end{tabular}
	\begin{tabular}{c|ccccc}
		\toprule
		\multirow{2}{*}{Threshold} & \multicolumn{5}{c}{$\alpha$} \\
		& $0.1$ & $0.5$ & $1$ & $10$ & $100$ \\
		\hline
		$\leq 0.1 \%$ & $\bm{100} \%$ & $\bm{100} \%$ &$\bm{100} \%$ &$20 \%$ &$0 \%$\\ 
		$\leq 0.5 \%$ & $\bm{100} \%$ & $\bm{100} \%$ &$\bm{100} \%$ &$\bm{100} \%$ &$\bm{100} \%$\\ 
		$\leq 1.0 \%$ & $\bm{100} \%$ & $\bm{100} \%$ &$\bm{100} \%$ &$\bm{100} \%$ &$\bm{100} \%$\\ 
		\bottomrule 
	\end{tabular}
	\Cblack
	\label{tab: ranking_confirmation}
\end{table}